%% file: Timeevo_Cert_20171212.tex
\pdfoutput=1

\newcommand\mynoinkscapesvgtopdf{}

\ifx\myrevtex\undefined

\providecommand{\mydocumentoptionsdefault}{bibliography=totoc,twocolumn=false,twoside=false,titlepage=false,parskip=half-,abstract=false,fleqn,draft=false,paper=a4,fontsize=10pt}

\else

\newcommand{\mydocumentoptions}{aip,jmp,footinbib,preprint}

\fi

\newcommand{\mypaper}{A4}\newcommand{\mydiv}{current}

\newcommand{\mydisablewikimode}{}
\newcommand{\mycustommaketitle}{1}

\input{aux/header-jabref}

\ifx\myrevtex\undefined
\fi

\usepackage[acronym,toc,section=section]{glossaries}

\makeatletter
\newcommand*{\glsplainhyperlink}[2]{%
  \colorlet{currenttext}{.}
  \colorlet{currentlink}{\@linkcolor}
  \hypersetup{linkcolor=currenttext}
  \hyperlink{#1}{#2}%
  \hypersetup{linkcolor=currentlink}
}
\let\@glslink\glsplainhyperlink
\makeatother

\newglossary[slg]{symbols}{sym}{sbl}{Symbols}
\makeglossaries

\input{aux/header-glossary}

\Crefname{prop}{Proposition}{Propositions}
\Crefname{lem}{Lemma}{Lemmata}
\Crefname{thm}{Theorem}{Theorems}
\Crefname{cor}{Corollary}{Corollaries}
\Crefname{rem}{Remark}{Remarks}

\newcommand{\affilulm}{Institut f\"ur Theoretische Physik and IQST, Albert-Einstein-Allee 11, Universit\"at Ulm, 89069 Ulm, Germany}

\ifx\myrevtex\undefined
\usepackage{authblk}

\renewenvironment{abstract}{}{}

\author[1]{Milan Holz\"apfel\thanks{\texttt{mail@mholzaepfel.de}}$^,$}
\author[1]{Martin B. Plenio}
\affil[1]{\affilulm}

\else

\newcommand{\parencite}[1]{\cite{#1}}

\makeatletter

\def\frontmatter@affiliationfont{%
 \frontmatter@@indent
 \preprintsty@sw{\small}{\small}%
 \it
}%
\def\frontmatter@title@format{%
 \preprintsty@sw{\large}{\Large}%
 \sffamily
 \bfseries
 \raggedright
 \color{Mahogany}
 \parskip\z@skip
}%

\def\thesection       {\arabic{section}}%
\def\p@section        {}%
\def\thesubsection    {\arabic{subsection}}%
\def\p@subsection     {\thesection.}%
\def\thesubsubsection {\arabic{subsubsection}}%
\def\p@subsubsection  {\thesection.\thesubsection.}%
\def\theparagraph     {\Alph{paragraph}}%
\def\p@paragraph      {\thesection.\thesubsection.\thesubsubsection\,}%
\def\p@subparagraph   {\thesection.\thesubsection.\thesubsubsection\,\theparagraph\,}%

\def\section{%
  \@startsection
    {section}%
    {1}%
    {\z@}%
    {0.8cm \@plus1ex \@minus .2ex}%
    {0.5cm}%
    {%
     \normalfont
     \bfseries
     \raggedright
     \color{Mahogany}
    }%
}%
\def\subsection{%
  \@startsection
    {subsection}%
    {2}%
    {\z@}%
    {.8cm \@plus1ex \@minus .2ex}%
    {.5cm}%
    {%
     \normalfont
     \bfseries
     \raggedright
     \color{Mahogany}
    }%
}%
\def\subsubsection{%
  \@startsection
    {subsubsection}%
    {3}%
    {\z@}%
    {.8cm \@plus1ex \@minus .2ex}%
    {.5cm}%
    {%
     \normalfont
     \itshape\bfseries
     \raggedright
     \color{Mahogany}
    }%
}%

\makeatother

\fi

\newcommand{\myabstract}{%
  We discuss efficient simulation and certification of the dynamics induced by a quantum many-body Hamiltonian $H$ with short-ranged interactions, extending prior results for one-dimensional systems
  [Osborne, Phys. Rev. Lett. 97, 157202 (2006) and Lanyon, Maier et al, Nat. Phys. 13, 1158 (2017)]
  to lattices in arbitrary spatial dimensions.
}

\counterwithout{equation}{section}
\counterwithout{theorem}{section}

\begin{document}


\ifx\myrevtex\undefined
\else
\author{Milan Holz\"apfel}
\email{mail@mholzaepfel.de}
\affiliation{\affilulm}
\author{Martin B. Plenio}
\affiliation{\affilulm}
\fi

\title{Efficient certification and simulation of local quantum many-body Hamiltonians}
\date{December 12, 2017}

\ifx\myrevtex\undefined

\maketitle
\begin{abstract}
  \myabstract
\end{abstract}

\else

\begin{abstract}
  \myabstract
\end{abstract}
\maketitle

\fi


\tocIfNoPreview
\vspace{2em}

\newcommand{\standalonedocument}{}

\input{c-Timeevo_Cert}

\section*{Acknowledgements}
We acknowledge discussions with Dario Egloff and Ish Dhand.
We acknowledge support from an Alexander von Humboldt Professorship, the ERC Synergy grant BioQ, the EU project QUCHIP, the US Army Research Office Grant No. W91-1NF-14-1-0133. 

\appendix
\counterwithin{equation}{section}
\counterwithin{theorem}{section}

\input{c-Appendix_TEC}

\printbibliography
\printglossary[type=\acronymtype]

\end{document}


%% file: aux/header-jabref.tex
\input{aux/header-all}

\ifx\myrevtex\undefined
\else
\newcommand{\printbibliography}{\bibliography{aux/quellen-jabref-tex.bib}}
\fi

%% file: aux/header-all.tex
\ifx\myrevtex\undefined

\providecommand{\mydocumentclass}{scrartcl}
\providecommand{\mydocumentoptions}{}
\providecommand{\mydocumentoptionsdefault}{twocolumn=false,twoside=false,titlepage=false,parskip=half-,abstract=false,fleqn,draft=false,paper=a4,fontsize=10pt}
\providecommand{\mypaper}{414.936062pt:586.838153pt}
\providecommand{\mydiv}{30}
\documentclass[\mydocumentoptionsdefault,\mydocumentoptions,paper=\mypaper,DIV=\mydiv]{\mydocumentclass}

\else

\providecommand{\mydocumentclass}{revtex4-1}
\providecommand{\mydocumentoptions}{}
\documentclass[\mydocumentoptions]{\mydocumentclass}

\fi

\input{aux/header-part1-usepackage}

\ifx\myrevtex\undefined
\providecommand{\MyBibLaTeXStyle}{authoryear-comp}
\usepackage[
  backend=biber,
style=\MyBibLaTeXStyle,
  natbib=false,
  sortcites=false,
]{biblatex}

\input{aux/header-part2-biboptions}

\fi

\input{aux/\mycommandsfile}

\ifx\mydisablewikimode\undefined

\input{aux/header-wiki}

\fi

%% file: aux/header-part1-usepackage.tex
\ifx\myrevtex\undefined

\fi

\usepackage{amsmath}
\usepackage{amssymb}

\usepackage[english]{babel}
\ifx\myrevtex\undefined
\usepackage[T1]{fontenc}
\usepackage{newpxtext}
\usepackage[amssymbols,vvarbb]{newpxmath}
\else
\usepackage[T1]{fontenc}
\usepackage{mathpazo}
\fi

\ifx\myrevtex\undefined

\KOMAoptions{paper=\mypaper,DIV=\mydiv}

\setkomafont{disposition}{\rmfamily\mdseries\color{Mahogany}}
\setkomafont{title}{\rmfamily\mdseries}
\setkomafont{paragraph}{\rmfamily\bfseries\normalcolor}

\fi

\usepackage[standard,thmmarks,hyperref]{ntheorem}

\usepackage[amssymb]{SIunits}

\ifx\myrevtex\undefined
\usepackage[format=plain,indention=0.5cm,font=small,margin={0cm,0cm,oneside}]{caption}
\else
\usepackage{capt-of}
\fi

\usepackage[utf8]{inputenc}
\usepackage{textcomp}

\DeclareUnicodeCharacter{00A0}{~}

\usepackage{graphicx}
\usepackage{tabularx}
\usepackage{dcolumn}
\newcolumntype{.}{D{.}{.}{7}}
\newcolumntype{h}{D{.}{.}{3}}
\newcolumntype{d}[1]{D{.}{.}{#1}}
\graphicspath{{graphics/}}

\usepackage[dvipsnames,table]{xcolor}

\usepackage{chngcntr}
\counterwithin{equation}{section}
\counterwithin{theorem}{section}

\ifx\myrevtex\undefined
\usepackage{scrpage2}
\usepackage{pdflscape}
\fi

\usepackage{multirow}

\usepackage{wrapfig}

\usepackage{csquotes}

\usepackage{afterpage}

\usepackage{soul}

\usepackage{varwidth}

\providecommand{\mycustommaketitle}{0}

\if \mycustommaketitle 1

\makeatletter
\renewcommand\maketitle{\par
  \begingroup%
  \newpage%
  \global\@topnum\z@   
  \null%
  \renewcommand{\thefootnote}{\fnsymbol{footnote}}%
  \begin{center}%
    {\LARGE\color{Mahogany} \@title \par}%
    \vskip 1.5em%
    {\large%
      \lineskip .5em%
      \begin{tabular}[t]{c}%
        \@author%
      \end{tabular}\par}%
    \vskip 1em%
    {\large \@date}%
  \end{center}%
  \par%
  \vskip 1.5em%
  \thispagestyle{plain}\@thanks%
  \setcounter{footnote}{0}
  \endgroup%
}
\makeatother

\fi

\usepackage[
    breaklinks,pdftex,
    colorlinks,
    linkcolor=black,citecolor=black,urlcolor=MidnightBlue,filecolor=black,
    hyperindex
]{hyperref}

\usepackage{nameref}
\usepackage{zref-xr}
\zxrsetup{tozreflabel=false,toltxlabel=true,verbose=true}
\usepackage[capitalise]{cleveref}
\Crefname{prop}{Proposition}{Propositions}
\Crefname{lem}{Lemma}{Lemmata}
\Crefname{thm}{Theorem}{Theorems}
\Crefname{cor}{Corollary}{Corollaries}



%% file: aux/header-part2-biboptions.tex
\providecommand{\mybiblatexgiveninits}{giveninits}

\DeclareFieldFormat
  [article,inbook,incollection,inproceedings,patent,thesis,unpublished]
  {title}{`#1'}

\renewbibmacro{in:}{%
  \ifentrytype{article}{%
  }{%
    \printtext{\bibstring{in}\intitlepunct}%
  }%
}

\DeclareFieldFormat[article]{volume}{\mkbibbold{#1}}
\DeclareFieldFormat[article]{number}{\mkbibparens{#1}}

\renewbibmacro*{volume+number+eid}{%
  \printfield{volume}%
  \printfield{number}%
  \setunit{\addcomma\space}%
  \printfield{eid}}

\DeclareNameAlias{last-first/first-last}{last-first}

\input{aux/header-titleurllink}

\DeclareFieldFormat{title}{\printtext[linked]{\mkbibemph{#1}}}
\DeclareFieldFormat
  [article,inbook,incollection,inproceedings,patent,thesis,unpublished]
  {title}{\mkbibquote{\printtext[linked]{#1}\isdot}}
\DeclareFieldFormat[suppbook,suppcollection,suppperiodical]{title}
  {\printtext[linked]{#1}}


%% file: aux/header-titleurllink.tex
\DeclareFieldFormat{linked}{%
  \ifboolexpr{ test {\ifhyperref} and not test {\ifentrytype{online}} }
    {\iffieldundef{doi}
       {\iffieldundef{url}
          {\iffieldundef{isbn}
             {\iffieldundef{issn}
                {#1}
                {\href{\worldcatsearch\thefield{issn}}{#1}}}
             {\href{\worldcatsearch\thefield{isbn}}{#1}}}
          {\href{\thefieldfirstword{url}}{#1}}}
       {\href{http://dx.doi.org/\thefield{doi}}{#1}}}
    {#1}}

\def\worldcatsearch{http://www.worldcat.org/search?qt=worldcat_org_all&q=}

\makeatletter
\def\thefieldfirstword#1{%
  \expandafter\expandafter
  \expandafter\firstword
  \expandafter\expandafter
  \expandafter{\csname abx@field@#1\endcsname}}
\def\firstword#1{\firstword@i#1 \@nil}
\def\firstword@i#1 #2\@nil{#1}
\makeatother

\DeclareFieldFormat{url}{\url{\firstword{#1}}}

\renewbibmacro*{title}{
  \ifboolexpr{ test {\iffieldundef{title}} and test {\iffieldundef{subtitle}} }
    {}
    {\printtext[title]{\printtext[linked]{%
       \printfield[titlecase]{title}%
       \setunit{\subtitlepunct}%
       \printfield[titlecase]{subtitle}}}%
     \newunit}%
  \printfield{titleaddon}}

\renewbibmacro*{periodical}{
  \iffieldundef{title}
    {}
    {\printtext[title]{\printtext[linked]{%
       \printfield[titlecase]{title}%
       \setunit{\subtitlepunct}%
       \printfield[titlecase]{subtitle}}}}}

%% file: aux/commands-v2.tex


\newcommand{\parena}[1]{{\left(#1\right)}}

\newcommand{\parenc}[1]{{\bigl(#1\bigr)}}

\newcommand{\sqba}[1]{{\left[#1\right]}}

\newcommand{\sqbd}[1]{{\Bigl[#1\Bigr]}}

\newcommand{\cur}[1]{{\{#1\}}}
\newcommand{\cura}[1]{{\left\{#1\right\}}}

\newcommand{\bra}[1]{{\langle#1\rvert}}

\newcommand{\ket}[1]{{\lvert#1\rangle}}

\newcommand{\braket}[2]{{\langle#1\vert#2\rangle}}

\newcommand{\ketbra}[2]{{\vert#1\rangle\!\langle#2\vert}}

\newcommand{\abs}[1]{{\lvert#1\rvert}}

\newcommand{\norm}[1]{{\lVert#1\rVert}}

\newcommand{\opnorm}[1]{{\lVert#1\rVert_{(\infty)}}}
\newcommand{\opnorma}[1]{{\left\lVert#1\right\rVert_{(\infty)}}}

\newcommand{\opnormc}[1]{{\bigl\lVert#1\bigr\rVert_{(\infty)}}}
\newcommand{\opnormd}[1]{{\Bigl\lVert#1\Bigr\rVert_{(\infty)}}}

\newcommand{\trnorm}[1]{{\lVert#1\rVert_{(1)}}}

\newcommand{\trnormc}[1]{{\bigl\lVert#1\bigr\rVert_{(1)}}}

\newcommand{\floor}[1]{{\lfloor#1\rfloor}}

\newcommand{\ceil}[1]{{\lceil#1\rceil}}
\newcommand{\ceila}[1]{{\left\lceil#1\right\rceil}}







\newcommand{\adjm}{*}




\newcommand{\floorofa}{{\dot a}}








\newcommand{\hilb}[1]{\mathcal H_{#1}}




\newcommand{\ns}{n}







\newcommand{\mymathbb}[1]{\mathbb{#1}}


\newcommand{\vd}{{\,\mathrm d}}
\newcommand{\ee}{{\mathrm e}}
\newcommand{\ii}{{\mathrm i}}


\DeclareMathOperator{\poly}{poly}
\newcommand{\bigo}[1]{\mathcal O(#1)}
\newcommand{\eqbigo}[1]{= \bigo{#1}}


\newcommand{\mci}{\mathcal I}
\newcommand{\mcy}{\mathcal Y}
\newcommand{\mcz}{\mathcal Z}
\newcommand{\dotcup}{\mathbin{\dot{\cup}}}
\newcommand{\bigdotcup}{\mathop{\dot{\bigcup}}}


\renewcommand\Re{\operatorname{Re}}

\newcommand{\idm}{\mathbb{1}}
\newcommand{\tr}{\operatorname{Tr}}

\DeclareMathOperator{\diam}{diam}
\DeclareMathOperator{\lsb}{LSB}



\newcommand{\Z}{{\mymathbb Z}}
\newcommand{\R}{{\mymathbb R}}


\newcommand{\HH}{{\mathcal H}}




\newenvironment{spmatrix}{%
  \left(%
  \begin{smallmatrix}%
}{%
  \end{smallmatrix}%
  \right)%
}













\newcommand{\executeiffilenewer}[3]{%
\ifnum\pdfstrcmp{\pdffilemoddate{#1}}%
{\pdffilemoddate{#2}}>0%
{\immediate\write18{#3}}\fi%
}

\ifx\mynoinkscapesvgtopdf\undefined


\newcommand{\isvgc}[2][]{%
\executeiffilenewer{#2.svg}{#2-isvgt_full.pdf}%
{/home/user/.config/inkscape/extensions/svgtextext_to_pdflatex.py -full-pdf #2.svg #2-isvgt.pdf%
}%
\includegraphics[#1]{#2-isvgt_full.pdf}%
}
\else


\newcommand{\isvgc}[2][]{%
\includegraphics[#1]{#2-isvgt_full.pdf}%
}

\fi




\providecommand{\theoremcounter}{theorem}

\theoremstyle{plain} 
\theoremheaderfont{\normalfont\bfseries}
\newtheorem{thm}[\theoremcounter]{Theorem}
\newtheorem{cor}[\theoremcounter]{Corollary}

\newtheorem{lem}[\theoremcounter]{Lemma}

\theorembodyfont{\upshape}



\theoremsymbol{\ensuremath{_{\Box}}}
\newtheorem{defn}[\theoremcounter]{Definition}
\renewtheorem{rem}[\theoremcounter]{Remark}

\theoremstyle{nonumberplain}
\theoremheaderfont{\bfseries}
\theorembodyfont{\normalfont}
\theoremsymbol{\ensuremath{_\blacksquare}}
\renewtheorem{proof}{Proof}
\renewtheorem{beweis}{Beweis}



\AtBeginDocument{

  \newif\ifPreviewTest
  \newif\ifPreviewX

  \ifx\ifPreview\undefined
    \PreviewXfalse
  \fi

  \PreviewTestfalse
  \ifx\ifPreview\ifPreviewTest
    \PreviewXfalse
  \fi

  \PreviewTesttrue
  \ifx\ifPreview\ifPreviewTest
    \PreviewXtrue
  \fi

}

\newcommand{\tocIfNoPreview}{
  \ifPreviewX
  \section*{Contents not available (preview active)}
  \else
  \tableofcontents
  \fi
}





%% file: aux/header-glossary.tex
\newacronym{iid}{i.i.d.}{independent, indentically distributed}
\newacronym{mse}{MSE}{mean squared error}
\newacronym{pdf}{pdf}{probability density function}
\newacronym{pmf}{pmf}{probability mass function}
\newacronym{umvu}{UMVU}{uniform minimum variance unbiased}
\newacronym{povm}{POVM}{positive operator-valued measure}
\newacronym{cp}{CP}{completely positive}
\newacronym{tp}{TP}{trace-preserving}
\newacronym{cptp}{CPTP}{completely positive and trace preserving}
\newacronym[longplural={reduced density matrices}]{rdm}{RDM}{reduced density matrix}
\newacronym{svd}{SVD}{singular value decomposition}
\newacronym{ic}{IC}{informationally complete}
\newacronym{mle}{MLE}{maximum likelihood estimation}
\newacronym{qst}{QST}{quantum state tomography}
\newacronym{aapt}{AAPT}{ancilla-assisted process tomography}

\newacronym{dmrg}{DMRG}{density matrix renormalization group}
\newacronym{mps}{MPS}{matrix product state}
\newacronym{peps}{PEPS}{projected entangled pair state}
\newacronym{pepo}{PEPO}{projected entangled pair operator}
\newacronym{sgs}{SGS}{sequentially generated state}
\newacronym{bsgs}{BSGS}{block sequentially generated state}
\newacronym{tns}{TNS}{tensor network state}
\newacronym{tt}{TT}{tensor train}

\newacronym{mpo}{MPO}{matrix product operator}
\newacronym{mpdo}{MPDO}{matrix product density operator}
\newacronym{pmps}{PMPS}{locally purified matrix product state}

\newacronym{dfs}{DFS}{depth-first search}

\newacronym{mbqc}{MBQC}{measurement-based quantum computation}

\newglossaryentry{tmode}{
  name={mode (of a tensor)},
  description={Index of a multi-dimensional array (tensor)},
}

\newglossaryentry{tensor}{
  name={tensor},
  description={Array with more than two dimensions (indices, modes)},
}

\newglossaryentry{imag-i}{
  type=symbols,
  sort=i,
  name={\ensuremath{\ii}},
  description = {Imaginary unit, $\ii^{2} = -1$},
}
\newglossaryentry{euler-e}{
  type=symbols,
  sort=e,
  name={\ensuremath{\ee}},
  description = {Euler's number, $\ee = 2.71\dots$},
}
\newglossaryentry{hilbert-H}{
  type=symbols,
  sort=H,
  name={\ensuremath{\HH}},
  description = {Hilbert space},
}
\newglossaryentry{ln}{
  type=symbols,
  sort=ln,
  name={\ensuremath{\ln}},
  description = {Natural logarithm},
}
\newglossaryentry{log-2}{
  type=symbols,
  sort=log,
  name={\ensuremath{\log}},
  description = {Base-2 logarithm},
}
\newglossaryentry{L-A}{
  type=symbols,
  sort=LA,
  name={\ensuremath{L(A, B)}},
  description = {Set of linear maps from $A$ to $B$},
}
\newglossaryentry{H-A}{
  type=symbols,
  sort=HA,
  name={\ensuremath{H(A)}},
  description = {Set of Hermitian linear operators on $A$},
}

\newacronym[plural=LEDs, longplural={light-emitting diodes}]%
{led}{LED}{light-emitting diode}

\newglossaryentry{culdesac}{name=cul-de-sac, %
description={passage or street closed at one end}, %
plural=culs-de-sac}

\newglossaryentry{ohm}{type=symbols, name={\ensuremath{\Omega}}, %
sort=Ohm, symbol={\ensuremath{\Omega}}, %
description={unit of electrical resistance}}

%% file: c-Timeevo_Cert.tex
\section{Summary%
  \label{sec:timeevo-summary}}

\glsreset{peps}
\glsreset{pepo}

In this contribution, we discuss efficient simulation and certification of the dynamics induced by a quantum many-body Hamiltonian $H$ with short-ranged interactions. Here, we extend prior results for one-dimensional systems \parencite{Osborne2006,Lanyon2016} to lattices in arbitrary spatial dimensions. 
The Hamiltonian acts on $n < \infty$ quantum systems arranged in an arbitrary lattice in an arbitrary spatial dimension. 
We consider Hamiltonians whose interactions have a strictly finite range.

A function $f(n)$ is quasi-polynomial in $n$ if $f(n) \eqbigo{\exp(c_{1} (\log n)^{c_{2}})} \eqbigo{n^{c_{1} (\log n)^{c_{2}-1}}}$ with constants $c_{1,2} > 0$.
A function is poly-logarithmic in $n$ if $f(n) \eqbigo{(\log n)^{c_{1}}}$. 

We present a method which can certify the fact that an unknown quantum system evolves according to a certain Hamiltonian.
Suppose that the evolution time grows at most poly-logarithmically with $n$. 
We prove that the necessary measurement effort scales quasi-polynomially in the number of particles $n$.
It also scales quasi-polynomially in the inverse tolerable error $1/\mci$. 

In addition, we show that a \gls{peps} representation of a time-evolved state can be obtained efficiently in the following sense.
Suppose that the the evolution time $t$ grows at most poly-logarithmically with $n$. 
We prove that the necessary computation time and the \gls{peps} bond dimension of the representation scale quasi-polynomially in the number of particles $n$ and the inverse approximation error $1/\epsilon$.

For certification of a time-evolved state, we consider an initial product state $\ket{\psi(0)}$, the time-evolved state $\ket{\psi(t)} = \exp(-\ii H t) \ket{\psi(0)}$ and an unknown state $\rho$.
We measure the distance between a pure and a mixed state by the infidelity \[I(\rho, \ket\psi) = 1 - \bra\psi \rho \ket\psi.\]
In order to \emph{certify} that the unknown state is $\rho$ is almost equal to the time-evolved state $\ket{\psi(t)}$, we provide an upper bound $\beta$ on the distance of the two states, i.e.\ \[I(\rho, \ket{\psi(t)}) \le \beta.\]%
\begin{figure}[t]%
  \centering%
  \isvgc{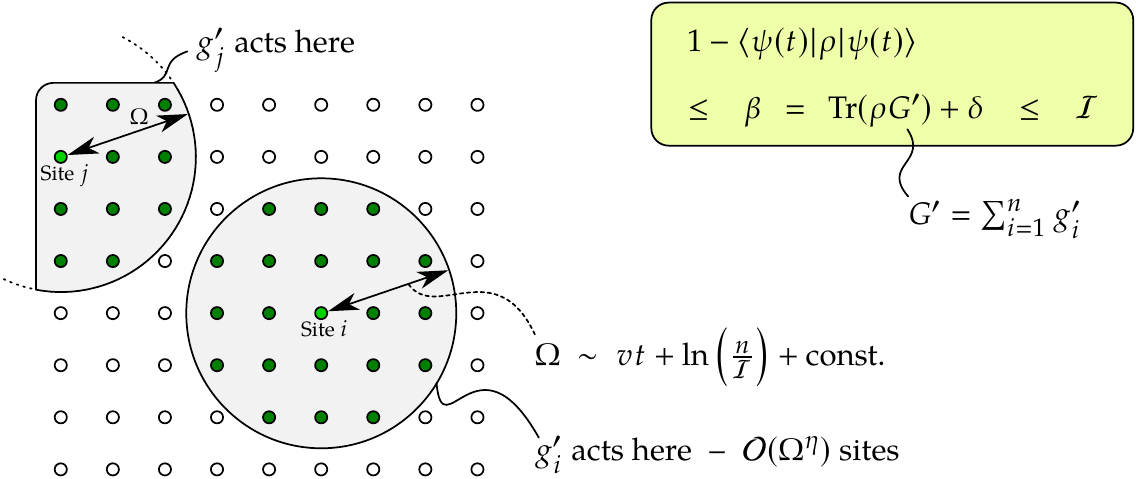}%
  \caption{%
    The local terms $g'_{i}$ of $G' = \sum_{i} g'_{i}$ act on regions whose diameter is proportional to $\Omega$, i.e.\ on $\bigo{\Omega^{\eta}}$ sites if the lattice has $\eta$ dimensions.
    The Lieb--Robinson velocity $v$ determines the growth of $\Omega$ with time. 
    The expectation value $\tr(\rho G')$, which provides an upper bound on the distance $I(\rho, \ket{\psi(t)}) = 1 - \bra{\psi(t)} \rho \ket{\psi(t)}$, can be determined from complete measurements on $n$ regions of size $\bigo{\Omega^{\eta}}$ sites (\cref{thm:time-evo-approx-parent-h,eq:cert-discuss-simple-D}). 
  }%
  \label{fig:cert-basics}%
\end{figure}%
We prove that the bound $\beta$ can be obtained from the expectation values of complete sets of observables on regions whose diameter is proportional to some number $\Omega$ (\cref{fig:cert-basics}).
If the unknown state $\rho$ is exactly equal to the time-evolved state $\ket{\psi(t)}$, then a bound $\beta$ which is no larger than a tolerable error $\mci$ can be obtained if $\Omega$ grows linearly with $\log(n / \mci)$ and if it also grows linearly with the evolution time $t$.
If we assume a spatial dimension $\eta \ge 1$, a region of diameter $\sim \Omega$ contains $\sim \Omega^{\eta}$ sites.
Since there are $n$ regions of diameter $\Omega$ and since $\bigo{\exp(c \Omega^{\eta})}$ observables are sufficient for a complete set on a single region, the total measurement effort is $\bigo{n \exp(c \log(n)^{\eta})} = \bigo{n^{1 + (c (\log n)^{\eta-1})}}$, i.e.\ it increases quasi-polynomially with $n$. 
This scaling reduces to polynomial in $n$ if the system is one-dimensional ($\eta = 1$).
In addition, we show that the upper bound $\beta$ increases only slightly if $\rho$ has a finite distance from $\ket{\psi(t)}$ or if the bound is obtained from expectation values which are not known exactly, e.g.\ due to a finite number of measurements per observable. 

\glsreset{mps}

Suppose that the Hamiltonian is a nearest-neighbour Hamiltonian in one spatial dimension and that the evolution time $t$ grows at most logarithmically with the number of particles $n$.
In this case, an approximate \gls{mps} representation of the time-evolved state $\ket{\psi(t)}$ can be obtained efficiently, i.e.\ the computational time grows at most polynomially with $n/\epsilon$ where $\epsilon$ is the approximation error \parencite{Osborne2006}.
\Glspl{peps} are a generalization of \glspl{mps} to higher spatial dimensions.
It has been demonstrated that \gls{mps}-based numerical algorithms for computing time evolution can be applied to \gls{peps} as well \parencite{Murg2007,Verstraete2008}.
However, the computational time required by these algorithms has not been determined in general. 
Here, we show that an approximate \gls{peps} representation of the time-evolved state $\ket{\psi(t)}$ can be obtained efficiently for poly-logarithmic times (in $n$).
Suppose that the evolution time $t$ grows at most poly-logarithmically with $n$ (i.e.\ $t \sim (\log n)^{c}$). 
We prove that the necessary computational time and the \gls{peps} bond dimension of the representation scale quasi-polynomially in the number of particles $n$ and the inverse approximation error $1/\epsilon$. 
Furthermore, we show that there is an efficient \gls{pepo} representation of the unitary evolution generated by the Hamiltonian.
This representation is structured in a way which guarantees efficient computation of expectation values of single-site observables in $\ket{\psi(t)}$, an operation which can be computationally difficult for a general \gls{peps}.\footnote{%
  Computing the expectation value of a single-site observable in an arbitrary \gls{peps} has been shown to be \#P-complete and it is widely assumed that a polynomial-time solution for such problems does not exist \parencite{Schuch2007}.
}

In \cref{sec:timeevo-lr}, existing Lieb--Robinson bounds are introduced and some corollaries are derived.
In \cref{sec:timeevo-parent-H-cert}, so-called parent Hamiltonians and their use as fidelity witnesses is introduced \parencite{Cramer2010}.
Parent Hamiltonians are then used to efficiently certify time-evolved states.
In \cref{sec:timeevo-repr}, we construct efficient representations of a unitary time evolution operator $U_{t}$.
The first two subsections discuss the Trotter decomposition and introduce \gls{peps}.
The remaining two subsections construct an efficient representation of $U_{t}$ for an arbitrary lattice and for a hypercubic lattice: In the special case, a representation with improved properties is achieved.
\Cref{sec:timeevo-discussion} concludes.

\section{Lieb--Robinson bounds
  \label{sec:timeevo-lr}}

Suppose that $H$ is a nearest-neighbour Hamiltonian on a lattice. 
The time evolution of an observable $A$ under a Hamiltonian $H$ is given by $\tau^{H}_{t}(A) = \ee^{\ii H t} A \ee^{-\ii H t}$ (assuming that $H$ is time-independent).
Even if $A$ acts non-trivially only on a small part of the system, $\tau^{H}_{t}(A)$ acts on the full system for any $t > 0$ because the exponential functions contain arbitrarily large powers of $H$.
(We shall assume that no part of the system is decoupled from the rest.)
However, $\tau^{H}_{t}(A)$ can be approximated by an observable which acts non-trivially on a small region around the original $A$.
The approximation error is exponentially small in the diameter of the region and the error remains constant if the diameter increases linearly with time (see also \cref{fig:timeevo-quasilocality} on \cpageref{fig:timeevo-quasilocality}). 
In this sense, information propagates at a finite velocity in a quantum lattice system.
A \emph{Lieb--Robinson bound} is an upper bound on the norm of the commutator $[\tau^{H}_{t}(A), B]$ and provides a means to bound the error of the named approximation.
The first bound on the commutator $[\tau^{H}_{t}(A), B]$ has been given by \textcite{Lieb1972} for a regular lattice.
More recently, these bounds have been extended to lattices described with graphs or metric spaces \parencite{Nachtergaele2006,Hastings2006a,Nachtergaele2006a}.
For interactions which decay exponentially (polynomially) with distance, Lieb--Robinson bounds have been proved which are exponentially (polynomially) small in distance \parencite{Hastings2006a};
here, the distance is between the regions on which $A$ and $B$ act non-trivially. 

The time-evolved observable $\tau^{H}_{t}(A)$ can be approximated by $\tau^{H'}_{t}(A)$ where the Hamiltonian $H'$ contains only the interaction terms which act on a given region $R$ of the system and this has been proven for a one-dimensional system by \textcite{Osborne2006}.
An explicit bound on the approximation error $\opnorm{\tau^{H}_{t}(A) - \tau^{H'}_{t}(A)}$ for a lattice with a metric has been given by \textcite{Barthel2012} for the case of a local Liouvillian evolution.
Their result is limited to interactions with a strictly finite range but this restriction also enables an explicit definition of all constants.
In the remainder of this \namecref{sec:timeevo-lr}, we introduce their result and derive corollaries used below.

Given two sets $A$ and $B$, the expression $A \subset B$ denotes the implication $x \in A \Rightarrow x \in B$ ($A$ is not required to be a strict subset of $B$).
The expression $C = A \dotcup B$ implies that $C = A \cup B$ and $A \cap B = \emptyset$.
The sets $\{B_{i} \colon i\}$ are a partition of the set $A$ if $A = \bigdotcup_{i} B_{i}$.
For a function $f(\ns)$, we write $f = \bigo{\poly(\ns)}$ if there is a polynomial $g(\ns)$ such that $f(\ns) \le g(\ns)$ for all suitable $n$ (e.g.\ $n \ge 1$ if $n$ is the number of particles).
We write $f \eqbigo{\exp(\ns)}$ if there are constants $c_{1}$, $c_{2}$ such that $f(\ns) \le c_{1} \exp(c_{2} \ns)$ holds for all $n$.
Given a linear map $U$, $U^{\adjm}$ denotes its Hermitian adjoint (conjugate transpose). 

The time evolution from time $s$ to time $t$ under a time-dependent Hamiltonian $H(t)$ is described by the unitary $U_{ts} = [U_{st}]^{\adjm}$ given by the unique solution of $\partial_{t} U_{ts} = -\ii H(t) U_{ts}$ where $U_{ss} = \idm$, $s, t \in \R$ and $H(t)$ is assumed to be continuous except for finitely many discontinuities in any finite interval.
The unitary satisfies $\partial_{s} U_{ts} = +\ii U_{ts} H(s)$ and, if $H$ is time-independent, it is given by $U_{ts} = \exp(-\ii H (t-s))$. 
To distinguish time evolutions under different Hamiltonians, we use the notation $U^{H}_{ts} = U_{ts}$. 
The time evolution of a pure state $\ket{\psi(s)}$ and a density matrix $\rho(s)$ are given by $\ket{\psi(t)} = U_{ts} \ket{\psi(s)}$ and $\rho(t) = \tau_{ts}^{H}(\rho(s)) = U_{ts} \, \rho(s) \, U_{st}$.
If we omit the second time argument $s$, it is equal to zero: $U_{t} = U_{t0}$ and $\tau^{H}_{t} = \tau^{H}_{t0}$.

We consider a system of $n < \infty$ sites and $\Lambda$ denotes the set of all sites.
Associated to each site $x \in \Lambda$, there is a Hilbert space $\HH_{x}$ of finite dimension $d(x) \ge 2$. 
We assume that there is a metric $d(x, y)$ on $\Lambda$.
The diameter of a set $X \subset \Lambda$ is given by $\diam(X) = \max_{x,y\in X}d(x, y)$.
Distances between sets are given by $d(x, Y) = \min_{y \in Y} d(x, y)$ and $d(X, Y) = \min_{x\in X, y\in Y}d(x, y)$ where $X, Y \subset \Lambda$. 
The Hamiltonians $H_{V}$ and $H$ of a subsystem $V \subset \Lambda$ and of the whole system, respectively, are given by
\begin{align}
  \label{eq:timeevo-lr-H}
  H_{V}
  &= \sum_{Z \subset V} h_{Z},
  &
  H
  &= H_{\Lambda}.
\end{align}
The local terms $h_{Z}(t)$ can be time-dependent but we often omit the time argument. 
At a given time, each local term $h_{Z}(t)$ is either zero or acts non-trivially at most on $Z$.
The maximal norm and range of the local terms are given by
\begin{align}
  \label{eq:timeevo-lr-J-a}
  J
  &= 2\sup_{t,Z \subset \Lambda} \opnorm{h_{Z}(t)},
  &
    a &= \sup_{Z\colon h_{Z} \ne 0} \diam(Z).
\end{align}
Terms which act non-trivially only on a single site, which may unduly enlarge the maximal norm $J$, can be eliminated from our discussion by employing a suitable interaction picture as described in \cref{sec:lr-remove-single-site-terms}. 
The maximal number of nearest neighbours is given by
\begin{align}
  \label{eq:timeevo-lr-Z}
  \mcz
  &= \max_{Z\colon h_{Z} \ne 0}
    \abs{\cur{ Z' \subset \Lambda \colon h_{Z'} \ne 0, Z' \cap Z \ne \emptyset }}.
\end{align}
This restricts the number of local terms in the Hamiltonian to $\abs{\cur{Z \subset \Lambda\colon h_{Z} \ne 0}} \le \mcz n = \bigo{n}$.\footnote{%
  Write $\cur{Z \subset \Lambda\colon h_{Z} \ne 0} = \bigcup_{x \in \Lambda} \{Z \subset \Lambda\colon h_{Z} \ne 0, x \in Z\}$.
}
The number of local terms at a certain distance $r$ is given by the number of elements in the set
\begin{align}
  \label{eq:timeevo-lr-R-n}
  R_{r,y} =
  \cura{ Z \subset \Lambda \colon h_{Z} \ne 0, d(y, Z) / a \in [r, r+1) }
\end{align}
and we assume that it is bounded by a power law:
\begin{align}
  \label{eq:timeevo-lr-abs-R-n}
  \abs{R_{r,y}}
  \le
  M r^{\kappa}
  \quad \forall \; y \in \Lambda, r \in \{0, 1, 2, \dots \},
\end{align}
where $M$ and $\kappa$ are constants.
A regular lattice in an Euclidean space of dimension~$\eta$ satisfies this condition with $\kappa = \eta - 1$.
\Cref{eq:timeevo-lr-abs-R-n} restricts the number of local terms $h_{Z} \ne 0$ within a certain distance in terms of the metric but the number of sites on which a local term may act remains unbounded.
We demand that this number of sites is bounded by a finite
\begin{align}
  \label{eq:timeevo-lr-sites-per-term}
  \mcy = \sup_{Z\colon h_{Z} \ne 0} \abs{Z}.
\end{align}
We assume that for each $x \in \Lambda$, there is a $Z \subset \Lambda$ with $x \in Z$ and $h_{Z} \ne 0$.
Together with \cref{eq:timeevo-lr-abs-R-n,eq:timeevo-lr-sites-per-term}, this assumption implies that $\abs{B^{o,c}_{r}(\{x\})} = \bigo{r^{\eta}}$ where $\eta = \kappa+1$, $x \in \Lambda$, $B^{o}_{r}(X) = \{ y \in \Lambda \colon d(X, y) < r \}$, $B^{c}_{r}(X) = \{ y \in \Lambda \colon d(X, y) \le r \}$ and $X \subset \Lambda$. 
The extension of a volume $V \subset \Lambda$ in terms of the Hamiltonian is given by
\begin{align}
  \label{eq:timeevo-lr-bar-V}
  \bar V = \bigcup_{\substack{Z \colon h_{Z} \ne 0,\\Z \cap V \ne \emptyset}} Z.
\end{align}
The following \namecref{thm:timeevo-quasilocality} has been shown by \textcite[Theorem~2]{Barthel2012} and they have called it \emph{quasilocality}:%

\begin{thm}
  \label{thm:timeevo-quasilocality}
  Let $a$, $\mcz$ and $J$ be finite and $t \in \R$.
  Let $Y \subset R \subset \Lambda$ and let $A$ act on $Y$ (\cref{fig:timeevo-quasilocality}).
  Let $d_{a} = d(Y, \Lambda \setminus R) / a$ and $\ceil{d_{a}} > 2\kappa + 1$. 
  Then
  \begin{align}
    \opnorm{ \tau^{H}_{t}(A) - \tau^{H_{\bar R}}_{t}(A) }
    \le
    \frac{2M}{\mcz} \, \opnorm{A} \, \ceil{d_{a}}^{\kappa} \, \exp(v\abs t - \ceil{d_{a}})
    \label{eq:timeevo-quasilocality}
  \end{align}
  holds.
  The Lieb--Robinson velocity is given by $v = J \mcz \exp(1)$.
\end{thm}

\begin{figure}[t]
  \centering
  \isvgc{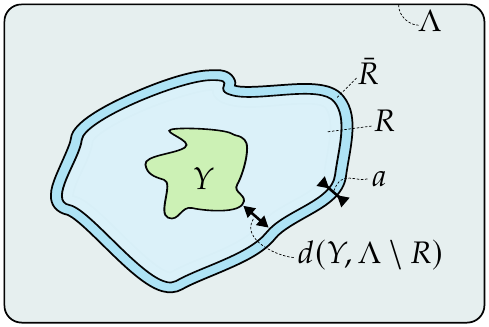}
  \caption{%
    An observable $A$ at time zero acts non-trivially on $Y$ and $\tau^{H}_{t}(A)$ is approximated by $\tau^{H_{\bar R}}_{t}(A)$ on $\bar R$.
    The approximation error is determined by the evolution time $t$ and the distance $d(Y, \Lambda \setminus R)$ (\cref{thm:timeevo-quasilocality}; \cite{Barthel2012}).
  }
  \label{fig:timeevo-quasilocality}
\end{figure}

The upper bound from \cref{eq:timeevo-quasilocality} can be simplified as $x^{\kappa} \exp(-x) \le \exp(-(1-q) x)$ holds for any $q \in (0, 1)$ if $x$ is large enough.
The following \namecref{lem:timeevo-scale-distance} provides a precise formulation of this fact and \cref{cor:timeevo-quasilocality} applies it to \cref{eq:timeevo-quasilocality}. 

\begin{lem}
  \label{lem:timeevo-scale-distance}
  
  Choose $q \in (0, 1)$ and set
  $
  D = (1-q) d_{a}
  $
  with $d_{a} \ge 0$.
  We say that $\ceil{d_{a}}$ is large enough if it satisfies $\ceil{d_{a}} > 2\kappa + 1$ and $\ceil{d_{a}} \ge \frac{2\kappa}{q} \ln(\frac{\kappa}{q})$; let $\ceil{d_{a}}$ be large enough.%
  \footnote{
    The slightly simpler conditions $D \ge 2\kappa + 1$ and $D \ge \frac{2\kappa}{q} \ln(\frac{\kappa}{q})$ are stricter and can also be used. 
  }
  Set $\alpha_{q} = \exp[-(1-q) (\ceil{d_{a}} - d_{a})]$.
  Then $\frac1\ee < \frac1{\exp(1-q)} < \alpha_{q} \le 1$ and $\ceil{d_{a}}^{\kappa} \exp(-\ceil{d_{a}}) \le \alpha_{q} \exp(-D)$ hold.
\end{lem}

\begin{proof}
  We have $\alpha_{q} = \exp(D - (1-q) \ceil{d_{a}})$. 
  Because $\ceil{d_{a}}$ was assumed to be large enough, we can use \cref{lem:timeevo-poly-exp-bound} to obtain
  $\ceil{d_{a}}^{\kappa} \exp(-\ceil{d_{a}}) \le \exp(-(1-q) \ceil{d_{a}}) =
    \alpha_{q} \exp(- D) $.
  This completes the proof.
\end{proof}

We simplify the bound from \cref{eq:timeevo-quasilocality} by applying \cref{lem:timeevo-scale-distance}:

\begin{cor}
  \label{cor:timeevo-quasilocality}

  Let $a$, $\mcz$ and $J$ be finite and $t \in \R$.
  Let $Y \subset R \subset \Lambda$ and let $A$ act on $Y$.
  Choose $q \in (0, 1)$ and set $D = (1-q) d_{a}$ where $d_{a} = d(Y, \Lambda \setminus R) / a$ and where $\ceil{d_{a}}$ is large enough (\cref{lem:timeevo-scale-distance}). 
  Then
  \begin{align}
    \opnorm{ \tau^{H}_{t}(A) - \tau^{H_{\bar R}}_{t}(A) }
    \le
    \frac{2M \alpha_{q}}{\mcz} \opnorm{A} \exp(v\abs t - D)
    \label{eq:cor-timeevo-quasilocality}
  \end{align}
  holds.
  The Lieb--Robinson velocity is given by $v = J \mcz \exp(1)$ and $\alpha_{q} \in (\exp(-(1-q)), 1]$.
  Specifically, $\alpha_{q} = \exp(-(1-q) (\ceil{d_{a}} - d_{a}))$.
\end{cor}

The upper bound from \cref{eq:cor-timeevo-quasilocality} is at most $\epsilon$ if $D$ is large enough:

\begin{cor}
  \label{cor:timeevo-quasilocality-D-scaling}
  
  Let $a$, $\mcz$ and $J$ be finite and $t \in \R$.
  Let $Y \subset R \subset \Lambda$ and let $A$ act on $Y$.
  Choose $q \in (0, 1)$ and set $D = (1-q) d_{a}$ where $d_{a} = d(Y, \Lambda \setminus R) / a$ and where $\ceil{d_{a}}$ is large enough (\cref{lem:timeevo-scale-distance}). 
  If $D$ satisfies
  \begin{align}
    D &\ge v\abs t + \ln\parena{ \frac{1}{\epsilon} } + \ln\parena{\opnorm{A}} + c_{1},
    &
      c_{1} &= \ln(2M/\mcz),
    \label{eq:timeevo-quasilocality-D-scaling}
  \end{align}
  then%
  \footnote{
    This holds for all $D$ which satisfy \eqref{eq:timeevo-quasilocality-D-scaling}; it holds e.g.\ if $D$ is equal to the lower bound stated in \eqref{eq:timeevo-quasilocality-D-scaling}.
  }
  \begin{align}
    \opnorm{ \tau^{H}_{t}(A) - \tau^{H_{\bar R}}_{t}(A) }
    \le
    \alpha_{q} \epsilon \le \epsilon.
  \end{align}
  The Lieb--Robinson velocity is given by $v = J \mcz \exp(1)$.
  Refer to \cref{cor:timeevo-quasilocality} for $\alpha_{q}$.
\end{cor}

\Cref{cor:timeevo-quasilocality-D-scaling} states that the time evolution $A(t)$ of a local observable $A(0)$ acting on $Y$ can be approximated by another local observable $A'(t)$ which acts on a certain region $\bar R$ around $Y$.
This is possible with high accuracy if the region $\bar R$ is large enough.
Suppose that $G$ is a sum of time-evolved local observables and $G'$ is obtained by taking the sum of corresponding approximated observables.
The next \namecref{lem:approx-loc-obs} compares the expectation value $\tr(\rho G')$ of the approximated observable $G'$ with the expectation values $\tr(\rho G)$ and $\tr(\psi G)$ where the quantum state $\psi$ has small distance from $\rho$ (in trace norm). 

\begin{lem}
  \label{lem:approx-loc-obs}
  
  Let $g_{i}(0)$ be observables with $\opnorm{g_{i}(0)} \le 1$ which act non-trivially on $Y_{i}$, $Y_{i} \subset R_{i} \subset \Lambda$.
  Choose a fixed time $t \in \R$ and let $G = \sum_{i=1}^{\Gamma} g_{i}(t)$ and $G'$ be the sum of $g_{i}'(t) = \tau_{t}^{H_{\bar R_{i}}}(g_{i}(0))$.
  Let $\psi$ and $\rho$ be quantum states with $\trnorm{\rho - \psi} \le \gamma$.
  Let $\Gamma \gamma < \mci$.
  Choose $q \in (0, 1)$ and set $D = (1-q) d_{a}$ where $d_{a} = \frac 1a \max_{i} d(Y_{i}, \Lambda \setminus R_{i})$.
  If $\ceil{d_{a}}$ is large enough (\cref{lem:timeevo-scale-distance}) and $D$ satisfies
  \begin{align}
    D &\ge v\abs t + \ln\parena{ \frac{2\Gamma}{\mci - \Gamma \gamma} } + c_{1},
  \end{align}
  then
  \begin{align}
    \tr(\rho G) \le \tr(\rho G') + \delta \le \tr(\psi G) + \mci
  \end{align}
  holds with $\delta = \frac12(\mci - \Gamma\gamma)$.
\end{lem}

\begin{proof}
  Set $\epsilon = \frac\delta\Gamma = \frac1{2\Gamma}(\mci - \Gamma\gamma)$. 
  Applying \cref{cor:timeevo-quasilocality-D-scaling} provides
  \begin{align}
    \opnorm{G - G'} \le
    \sum_{i=1}^{\Gamma}
    \opnorm{g_{i}(t) - g_{i}'(t)}
    \le
    \Gamma\epsilon =
    \delta.
    \label{eq:approx-loc-obs-G-Gp}
  \end{align}
  Using $\abs{\tr(\rho(G - G'))} \le \trnorm{\rho} \opnorm{G - G'}$ \parencite[Exercise~IV.2.12]{Bhatia1997} provides
  \begin{align}
    \tr(\rho G) \le \tr(\rho G') + \delta \le \tr(\rho G) + 2\delta.
    \label{eq:approx-loc-obs-G-Gp-2}
  \end{align}
  Using $\abs{\tr{(\rho - \psi) G}} \le \trnorm{\rho - \psi} \opnorm{G}$ provides
  \begin{align}
    \tr(\rho G) - \tr(\psi G) \le \trnorm{\rho - \psi} \opnorm{G} \le \gamma \Gamma.
    \label{eq:approx-loc-obs-rho-psi}
  \end{align}
  Inserting \eqref{eq:approx-loc-obs-rho-psi} into \eqref{eq:approx-loc-obs-G-Gp-2} completes the proof.
\end{proof}

\section{Efficient certification
  \label{sec:timeevo-parent-H-cert}}

An observable $G$ is called a parent Hamiltonian of a pure state $\ket\psi$ if $\ket\psi$ is a ground state of $G$ (i.e.\ an eigenvector of $G$'s smallest eigenvalue).
If such a ground state is non-degenerate, the expectation value $\tr(\rho G)$ in an arbitrary state $\rho$ provides a lower bound on the fidelity of $\rho$ and the ground state $\ket\psi$ \parencite{Cramer2010}:

\begin{lem}
  \label{lem:timeevo-parent-h-cert}
  
  Let $G$ be an observable with the two smallest eigenvalues $E_{0}$ and $E_{1} > E_{0}$.
  Let $\ket\psi$ be an eigenvector of the smallest eigenvalue $E_{0}$ and let $E_{0}$ be non-degenerate.
  Let $\rho$ be some quantum state. 
  Then,
  \begin{align}
    \label{eq:timeevo-parent-h-cert}
    1 - \bra\psi \rho \ket\psi
    \le
    \beta =
    \frac{E_{\rho} - E_{0}}{E_{1} - E_{0}}
  \end{align}
  where $E_{\rho} = \tr(\rho G)$ \parencite{Cramer2010}.
  The value of the right hand side is bounded by
  \begin{align}
    \label{eq:timeevo-parent-h-cert-worstcase}
    \beta = \frac{E_{\rho} - E_{0}}{E_{1} - E_{0}}
    \le
    \frac{\trnormc{\rho - \ketbra\psi\psi} \opnormc{G}}{E_{1} - E_{0}}.
  \end{align}
\end{lem}

\begin{proof}
  Proofs of \cref{eq:timeevo-parent-h-cert} have been given by \textcite{Cramer2010,BaumgratzPhD2014}.
  \Cref{eq:timeevo-parent-h-cert-worstcase} follows from
  \begin{align}
    \tr(\rho G) - E_{0} = \abs{\tr([\rho - \psi] G)} \le \trnorm{\rho - \psi} \opnorm{G}
    \label{eq:infid-energy-worstcase}
  \end{align}
  where $\psi = \ketbra\psi\psi$.
  In the second inequality, we have used \parencite[Exercise~IV.2.12]{Bhatia1997} and this completes the proof.
\end{proof}

\begin{rem}
  \label{eq:timeevo-parent-h-nonideal}

  Suppose that the expectation value $E_{\rho} = \tr(\rho G)$ is not exactly known e.g.\ because it has been estimated from a finite number of measurements.
  The resulting uncertainty about the value of $\beta$ is given by the uncertainty about $E_{\rho}$ multiplied by the inverse of the energy gap $\Delta = E_{1} - E_{0}$ above the ground state.
  For robust certification, this energy gap must be sufficiently large.

  Suppose that $\rho$ is the unknown quantum state of some experiment which attempts to prepare the state $\ket\psi$. 
  If the experiment succeeds, $\rho$ will be close to the ideal state $\ket\psi$ (e.g.\ in trace distance) but the two states will not be equal.
  The maximal value of the infidelity upper bound $\beta$ from \cref{eq:timeevo-parent-h-cert} is provided by \eqref{eq:timeevo-parent-h-cert-worstcase}.
  In the worst case, $\beta$ is given by the trace distance of $\rho$ and $\ket\psi$, multiplied by the ratio of the Hamiltonian's largest eigenvalue and its energy gap $\Delta$.

  In a typical application, the expectation value $E_{\rho}$ is not exactly known and the states $\rho$ and $\ket\psi$ are not exactly equal.
  In order to obtain a useful certificate, it is necessary that both the energy gap $\Delta$ is sufficiently large and that the largest eigenvalue $\opnorm{G}$ is sufficiently small.
\end{rem}

The following simple Lemma shows that pure product states admit a parent Hamiltonian that has unit gap and only single-site local terms.
This result is a simple special case of prior work involving matrix product states \parencite{Perez2006,Cramer2010,BaumgratzPhD2014}.

\begin{lem}
  \label{lem:product-state-parent-h}
  Let $\ket{\phi} = \ket{\phi_{1}} \otimes \ket{\phi_{2}} \otimes \cdots \otimes \ket{\phi_{n}}$ be a product state on $n$ systems of dimension $d_{i} \ge 2$, $\braket{\phi_{i}}{\phi_{i}} = 1$, $i \in \{ 1 \ldots n \}$.
  Define
  \begin{align}
    G = \sum_{i=1}^{n} h_{i}, \quad h_{i} = 
    \idm_{1, \dots, i-1} \otimes P_{\ker(\rho_{i})} \otimes \idm_{i+1, \dots, n}
  \end{align}
  where $P_{\ker(\rho_{i})} = \idm - \ketbra{\phi_{i}}{\phi_{i}}$ is the orthogonal projection onto the null space of the reduced density operator $\rho_{i} = \ketbra{\phi_{i}}{\phi_{i}}$ of $\ket\phi$ on site $i$.
  The eigenvalues of $G$ are given by $\{0, 1, 2, \dots, n\}$, the smallest eigenvalue zero is non-degenerate and $\ket\phi$ is an eigenvector of eigenvalue zero.
\end{lem}

\begin{proof}
  Let $\ket{\mu^{(i)}_{i_{k}}}$ ($i_{k} \in \{1, \dots, d_{i}\}$) an orthonormal basis of system $i$ with $\ket{\mu^{(i)}_{1}} = \ket{\phi_{i}}$ ($i \in \{1, \dots, n\}$).
  The product basis constructed from these bases is an eigenbasis of $H$:
  \begin{align*}
    G \ket\mu = \lambda \ket\mu, \quad \ket\mu = \ket{\mu^{(1)}_{i_{1}}} \otimes \dots \otimes \ket{\mu^{(n)}_{i_{n}}},
    \quad
    \lambda = \abs{\cur{k \in \{1\dots n\} \colon i_{k} > 1}}.
  \end{align*}
  As we required $d_{i} \ge 2$, the eigenvalues of $H$ are given by $\{0, 1, 2, \ldots, n\}$.
  We also see that the smallest eigenvalue zero is non-degenerate and $\ket{\phi}$ is an eigenvector of eigenvalue zero.
  This completes the proof.
\end{proof}

In \cref{lem:product-state-parent-h}, a parent Hamiltonian $G$ is constructed from projectors onto null spaces of single-site reduced density matrices.
One projection is required for each of the $n$ sites and this determines the value of the operator norm $\opnorm G = n$.
In \cref{lem:timeevo-parent-h-cert}, a smaller operator norm was seen to be advantageous for robust certification. 
By projecting onto null spaces of multi-site reduced density matrices, the following Lemma obtains a parent Hamiltonian with smaller operator norm.
More importantly, it also provides a parent Hamiltonian for the time-evolved state $\ket{\psi(t)}$.

\begin{lem}
  \label{lem:time-evo-parent-h}
  Let $\ket{\psi(0)} = \ket{\phi_{1}} \otimes \dots \otimes \ket{\phi_{n}}$ be a product state on the lattice $\Lambda$.
  Let $Y_{1}, \dots, Y_{\Gamma}$ a partition of the set of sites $\Lambda$. 
  For a subset $Y \subset \Lambda$, define $\ket{\phi_{Y}} = \bigotimes_{k \in Y} \ket{\phi_{k}}$. 
  Set $g_{i}(0) = (\idm - \ketbra{\phi_{Y_{i}}}{\phi_{Y_{i}}}) \otimes \idm_{\Lambda \setminus Y_{i}}$.
  Choose a fixed time $t \in \R$ and let $g_{i}(t) = \tau_{t}^{H}(g_{i}(0))$ and $G = \sum_{i=1}^{\Gamma} g_{i}(t)$.

  The time-evolved state $\ket{\psi(t)} = U_{t0} \ket{\psi(0)}$ is an eigenvector of $G$'s non-degenerate eigenvalue zero and the eigenvalues of $G$ are given by $\{0, 1, \dots, \Gamma\}$.
  
\end{lem}

\begin{proof}
  Let $G_{0} = \sum_{i=1}^{\Gamma} g_{i}(0)$.
  $G_{0}$'s eigenvalues are given by $\{0, \dots, \Gamma\}$ and $\ket{\psi(0)}$ is a non-degenerate eigenvector of $G_{0}$'s eigenvalue zero (\cref{lem:product-state-parent-h}; group sites into supersites as specified by the sets $Y_{i}$).
  The operators $G$ and $G_{0}$ are related by the unitary transformation $G = U_{t0} G_{0} U_{0t}$, which implies that they have the same eigenvalues including degeneracies and also that $G \ket{\psi(t)} = 0$.
  This completes the proof. 
\end{proof}

The parent Hamiltonian of $\ket{\psi(t)}$ from the last Lemma is not directly useful for certification because it is a sum of terms $g_{i}(t)$ which all act on the full system (for $t \ne 0$).
However, these terms can be approximated by terms which act on smaller regions, as described in the next \namecref{thm:time-evo-approx-parent-h}.
The \namecref{thm:time-evo-approx-parent-h} is illustrated in \cref{fig:cert-basics} on \cpageref{fig:cert-basics} for $Y_{i} = \{i\}$ ($i \in \Lambda$).

\begin{thm}
  \label{thm:time-evo-approx-parent-h}

  Consider the setting of \cref{lem:time-evo-parent-h} which includes a fixed time $t \in \R$.
  Choose sets $R_{i}$ such that $Y_{i} \subset R_{i} \subset \Lambda$ and let $G'$ be the sum of $g_{i}'(t) = \tau_{t}^{H_{\bar R_{i}}}(g_{i}(0))$.
  Let $\psi(t) = \ketbra{\psi(t)}{\psi(t)}$ and let $\rho$ be a state with $\trnorm{\rho - \psi(t)} \le \gamma$ (for the single chosen value of $t$). 
  Choose $\mci > 0$ such that
  \begin{align}
    \delta = \frac{\mci - \Gamma\gamma}2 > 0.
  \end{align}
  Set $D = (1-q) d_{a}$ where $d_{a} = \frac 1a \min_{i} d(Y_{i}, \Lambda \setminus R_{i})$ and $q \in (0, 1)$.
  If $\ceil{d_{a}}$ is large enough (\cref{lem:timeevo-scale-distance}) and $D$ satisfies
  \begin{align}
    D &\ge v \abs t + \ln\parena{ \frac{2\Gamma}{\mci - \Gamma \gamma} } + c_{1},
    &
      c_{1} &= \ln(2M/\mcz),
    \label{eq:time-evo-approx-parent-h-D-scaling}
  \end{align}
  then
  \begin{align}
    I(\rho, \psi(t))
    \le
    \tr(\rho G') + \delta
    \le
    \mci
    \label{eq:time-evo-approx-parent-h-fidelity-bound}
  \end{align}
  where $I(\rho, \psi(t)) = 1 - \bra{\psi(t)} \rho \ket{\psi(t)}$.
\end{thm}

\begin{proof}
  Using \cref{lem:timeevo-parent-h-cert}, the properties of $G$ from \cref{lem:time-evo-parent-h} imply that
  \begin{align}
    1 - \bra{\psi(t)} \rho \ket{\psi(t)}
    \le
    \tr(\rho G).
  \end{align}
  Inserting $G \ket{\psi(t)} = 0$, \cref{lem:approx-loc-obs} completes the proof. 
\end{proof}

The next \namecref{lem:time-evo-approx-cond} simplifies the premise of \cref{thm:time-evo-approx-parent-h} by eliminating $\Gamma$:

\begin{lem}
  \label{lem:time-evo-approx-cond}

  Let $f(D)$ a function with $1 \le f(D) \le \frac n\Gamma$.
  Assuming $n \gamma < \mci$, the inequality
  \begin{align}
    D + \ln(f(D)) \ge v \abs t + \ln\parena{\frac{2n}{\mci - n \gamma}} + c_{1}
    \label{eq:time-evo-approx-cond}
  \end{align}
  is sufficient for \eqref{eq:time-evo-approx-parent-h-D-scaling}.
\end{lem}

\begin{proof}
  The premise implies $\mci - n \gamma \le \mci - \Gamma \gamma$ and
  \begin{align}
    D \ge v \abs t + \ln\parena{\frac{2\Gamma}{\mci - n \gamma}} + c_{1}
    \ge v \abs t + \ln\parena{\frac{2\Gamma}{\mci - \Gamma \gamma}} + c_{1},
  \end{align}
  which completes the proof. 
\end{proof}

The following \namecref{lem:time-evo-approx-finite-meas} bounds the measurement effort if $\tr(\rho G')$ is estimated from finitely many measurements:

\begin{lem}
  \label{lem:time-evo-approx-finite-meas}

  Let $R = \max_{i} \abs{\bar R_{i}}$ be the maximal number of sites on which any of the local terms of $G'$ from \cref{thm:time-evo-approx-parent-h} act.
  On each region $\bar R_{i}$, choose an \gls{ic} \gls{povm} (examples are provided in \cref{rem:ic-povm-examples}).
  Let ``one measurement'' refer to one outcome of one of the \glspl{povm}.
  The upper bound $\tr(\rho G')+\delta$ from \cref{eq:time-evo-approx-parent-h-fidelity-bound} can be estimated with standard error $\epsilon$ from $M = \bigo{\exp(R) n^{3} / \epsilon^{2}}$ such measurements. 
\end{lem}

\begin{proof}
  The individual $\tr(\rho g'_{i}(t))$ can be estimated independently by carrying out separate measurements for the estimation of each $\tr(\rho g'_{i}(t))$.
  By the central limit theorem, $M'$ measurements are sufficient to estimate a single $\tr(\rho g'_{i}(t))$ with standard error $\epsilon' = c / \sqrt{M'}$.
  Here, $c \le \exp(\tilde c R) = \bigo{\exp(R)}$ where $\tilde c$ is a constant.
  To achieve standard error $\epsilon$ for $\tr(\rho G')$, we set $\epsilon' = \epsilon/n$ and obtain $M' = c^{2} n^{2}/\epsilon^{2}$.
  As separate measurements for each $g'_{i}(t)$ were assumed, the total number of measurements is at most $M = nM' = c^{2} n^{3}/\epsilon^{2}$. 
\end{proof}

\begin{rem}[Discussion of \cref{thm:time-evo-approx-parent-h}]
  \label{rem:time-evo-approx-parent-h-discuss}

  \Cref{thm:time-evo-approx-parent-h} provides a means to verify that an unknown state $\rho$ is close to an ideal time-evolved state $\psi(t)$ with the expectation values of few observables.
  Specifically, the \namecref{thm:time-evo-approx-parent-h} warrants that the infidelity $I(\rho, \psi(t))$ is at most $\beta = \tr(\rho G') + \delta$ where $G'$ is a sum of observables which act non-trivially only on small parts of the full system.
  Furthermore, the \namecref{thm:time-evo-approx-parent-h} guarantees $\beta \le \mci$ and we can choose any desired $\mci > 0$.
  To simplify the discussion, we restrict to $\trnorm{\rho - \psi(t)} \le \gamma = \frac{\mci}{2n}$: For larger systems or smaller certified infidelities, the unknown state $\rho$ must be closer to the ideal state $\psi(t)$.

  Let $\bigdotcup_{i=1}^{\Gamma} Y_{i} = \Lambda$ be a partition of $\Lambda$ with $\diam(Y_{i}) \le r'$ for some $r' > 0$.
  Note that $Y_{i} \subset B^{c}_{r'}(\{y_{i}\})$ holds for all $y_{i} \in Y_{i}$ (\cref{lem:metric-set-diam-closed-ball}). 
  Let $r > 0$ and set $R_{i} = B^{o}_{r}(Y_{i})$, then $\bar R_{i} \subset B^{o}_{r+r'+a}(\{y_{i}\})$ (\cref{lem:finite-metric-space}).
  We assume $r' = \bigo{r}$ and obtain $\abs{\bar R_{i}} = \bigo{(r+r'+a)^{\eta}} = \bigo{r^{\eta}}$ (because $a$ is independent of $n$).
  Note that $d(Y_{i}, \Lambda \setminus R_{i}) \ge r$ (\cref{lem:finite-metric-space}), i.e.\ $D \ge (1-q)r/a$ where $q \in (0, 1)$ is a constant.

  A particularly simple partition which works for any lattice is $Y_{i} = \{i\}$ with $i \in \Lambda = \{1\dots n\}$, $\Gamma = n$, $r' = 0$ and $\abs{\bar R_{i}} = \bigo{D^{\eta}}$.
  We choose $D$ according to \cref{lem:time-evo-approx-cond} using $f(D) = n / \Gamma = 1$:
  \begin{align}
    \label{eq:cert-discuss-simple-D}
    D = v \abs t + \ln\parena{\frac{4n}{\mci}} + c_{1}.
  \end{align}
  The length scale $D$ grows linearly in time and logarithmically in $n / \mci$.  
  As discussed in \cref{lem:time-evo-approx-finite-meas}, the measurement effort to estimate $\tr(\rho G')$ with standard error $\epsilon$ is
  \begin{align}
    \label{eq:cert-discuss-meas}
    \bigo{n^{3}\exp(D^{\eta})/\epsilon^{2}} =
    \mathcal O\parena{n^{3} \exp\parena{ \sqba{ v\abs t + \ln\parena{\frac{4n}{\mci}} + c_{1} }^{\eta}} /\epsilon^{2}}.
  \end{align}
  The measurement effort grows exponentially with time but only quasipolynomially with $n$ and with $\frac 1\mci$.
  For one-dimensional systems, $\eta = 1$, this quasipolynomial scaling reduces to a polynomial scaling.

  Finally, we explore what can be gained by choosing a coarser partition $\Lambda = Y_{1} \dotcup \dots \dotcup Y_{\Gamma}$ of a cubic lattice $\Lambda = \{1\dots L\}^{\eta}$ with the metric $d(x, y) = \max_{i}\abs{x_{i} - y_{i}}$.\footnote{%
    Cubic lattices are also discussed in more detail in \cref{sec:timeevo-repr-cubic-lattice}. 
  }
  Let $\Omega = \floor{Da} \in \{1 \dots L\}$ and $B = \ceil{L / \Omega}$.
  The cubic lattice can be divided into $\Gamma = B^{\eta}$ smaller cubes of maximal diameter $r' = \Omega = \bigo{D}$ and we still have $\abs{\bar R_{i}} = \bigo{(r+r'+a)^{\eta}} = \bigo{D^{\eta}}$.
  We set $f(D) = (\floor{Da}/2)^{\eta}$ which satisfies $f(D) \le n/\Gamma$.\footnote{%
    $\Omega \le L$ implies $L / \ceil{\frac L\Omega} > L/(\frac L\Omega + 1) = \Omega/(1+\frac\Omega L) \ge \Omega/2$, i.e.\ $n/\Gamma = (L/\ceil{\frac L\Omega})^{\eta} > (\Omega/2)^{\eta} = f(D)$. 
  }
  Inserting $f(D)$ into \cref{eq:time-evo-approx-cond} provides
  \begin{align}
    \label{eq:cert-discuss-improved-D}
    D + \eta \ln(\floor{Da}) \ge v \abs t + \ln\parena{\frac{4n}{\mci}} + c_{1} + \ln(2).
  \end{align}
  We have increased the radius of $\bar R_{i}$ from $(D + 1)a$ to about $(2D + 1)a$.
  However, the last equation shows that it is then already sufficient if $D$ grows slightly less than linearly in the right hand side, i.e.\ slightly less than mentioned above, as described by the additional logarithmic term.
\end{rem}

\begin{rem}[Examples of \gls{ic} \glspl{povm}]
  \label{rem:ic-povm-examples}

  In this remark, we discuss measurements on a region $\bar R_{i}$ where $i$ is fixed. 
  Recall that a set of operators $\{M_{k} \colon k\}$ on the Hilbert space $\hilb{\bar R_{i}}$ is a \gls{povm} if each $M_{k}$ is positive semidefinite and $\sum_{k} M_{k} = \idm$ \parencite[e.g.]{Nielsen2007}.
  The \gls{povm} is \gls{ic} if the operators $M_{k}$ span $\hilb{\bar R_{i}}$.

  Measurement outcomes of an \gls{ic} \gls{povm} on $R_{i}$ can be obtained in several different ways in an experiment.
  For example, a measurement of a tensor product observable $A = A_{1} \otimes \dots \otimes A_{\abs{\bar R_{i}}}$ on $\hilb{\bar R_{i}}$ returns one of the eigenvalues of $A$ as measurement outcome.
  Access to measurement outcomes of a set of observables which spans $\hilb{\bar R_{i}}$ allows sampling outcomes of an \gls{ic} \gls{povm} on $\bar R_{i}$.
  Alternatively, one can measure the single-site observables $A_{j}$ ($j \in \{1 \dots \abs{\bar R_{i}}\}$) in any order or simultaneously.
  Here, the measurement outcome is given by a vector $(\lambda_{1}, \dots, \lambda_{\abs{\bar R_{}}})$ where $\lambda_{j}$ is an eigenvalue of $A_{j}$.
  Access to this type of measurement outcomes of a set of observables which spans $\hilb{\bar R_{i}}$ provides another way to sample outcomes of an \gls{ic} \gls{povm} on $\bar R_{i}$.
\end{rem}

\Cref{lem:time-evo-parent-h} provides a parent Hamiltonian $G$ of the time-evolved state $\ket{\psi(t)}$ at a fixed time $t$.
\Cref{thm:time-evo-approx-parent-h} provides an upper bound on the distance between an unknown state and the time-evolved state in terms of $G'$ which is an approximation of $G$.
The next Lemma shows that $G'$ is the parent Hamiltonian of a state $\ket{\psi'}$ which is approximately equal to the time-evolved state.
As a consequence, an upper bound on the distance between an unknown state and $\ket{\psi'}$ can also be obtained.

\begin{lem}
  \label{lem:timeevo-approx-parent-h-cert}
  
  In the setting of \cref{thm:time-evo-approx-parent-h}, let
  \begin{align}
    \delta'
    &= \frac{\mci - \gamma\Gamma}{2(1+\mci)},
      &
    0 &< \delta' < \frac12.
  \end{align}
  Let the length $D$ be at least
  \begin{align}
    D \ge v \abs t + \ln\parena{ \frac{\Gamma}{\delta'} } + c_{1}.
  \end{align}
  The operator $G'$ has a non-degenerate ground state and the difference between its two smallest eigenvalues is at least $E_{1}' - E_{0}' \ge 1 - 2\delta'$.
  The ground state $\ket{\psi'}$ of $G'$ satisfies
  \begin{align}
    \abs{\braket{\psi(t)}{\psi'}}
    &\ge
      1 - \frac{\delta'}{1 - \delta'}
      \ge
      1 - 2\delta',
    &
    \trnorm{\psi(t) - \psi'}
    &\le 2 \sqrt{\frac{2\delta'}{1 - \delta'}} \le 4 \sqrt{\delta'}.
  \end{align}
  where $\psi(t) = \ketbra{\psi(t)}{\psi(t)}$ and $\psi' = \ketbra{\psi'}{\psi'}$. 
  For an arbitrary state $\rho$, the following inequality holds:
  \begin{align}
    1 - \bra{\psi'} \rho \ket{\psi'}
    \le
    \frac{\tr(\rho G') + \delta'}{1 - 2\delta'}
    \le
    \mci.
  \end{align}
\end{lem}

\begin{proof}
  Set $\epsilon = \frac{\delta'}{\Gamma}$.
  Applying \cref{cor:timeevo-quasilocality-D-scaling} provides
  \begin{align}
    \opnorm{G - G'} \le
    \sum_{i=1}^{\Gamma}
    \opnorm{g_{i}(t) - g_{i}'(t)}
    \le
    \Gamma\epsilon =
    \delta'.
  \end{align}
  All eigenvalues change by at most $\opnorm{G - G'} \le \delta'$ \parencite[Theorem~VI.2.1]{Bhatia1997}.
  Accordingly, the two smallest eigenvalues of $G'$ satisfy $E_{0}' \in [-\delta', \delta']$, $E_{1}' \in [1-\delta', 1+\delta']$ and $\delta' < \frac12$ ensures that the ground state remains non-degenerate.
  In addition, we have $E_{1}' - E_{0}' \ge 1 - 2\delta$.
  \Cref{lem:timeevo-parent-h-cert} provides
  \begin{align}
    1 - \bra{\psi'}\rho\ket{\psi'} \le \frac{\tr(\rho G') + \delta'}{1 - 2\delta'}.
  \end{align}
  We bound (cf.\ proof of \cref{lem:approx-loc-obs})
  \begin{align}
    \tr(\rho G') \le \tr(\rho G) + \delta'
    \le \trnorm{\rho - \psi(t)} \opnorm{G} + \delta'
    \le
    \gamma \Gamma + \delta'.
  \end{align}
  Combining the last two equations provides
  \begin{align}
    1 - \bra{\psi'}\rho\ket{\psi'}
    \le \frac{\tr(\rho G') + \delta'}{1 - 2\delta'}
    \le \frac{\gamma\Gamma + 2\delta'}{1 - 2\delta'}
    =
    \frac{  (1 + \mci) \gamma\Gamma + \mci - \gamma\Gamma  }
    {  1 + \mci - \mci + \gamma\Gamma  }
    = \mci.
  \end{align}

  To quantify the change in the ground state, we use \parencite[Theorem~VII.3.1]{Bhatia1997}
  \begin{align}
    \opnorm{EF} \le \frac{1}{\Delta} \opnorm{G - G'}
  \end{align}
  where $E = P_{G}(S_{1})$ and $F = P_{G'}(S_{2})$ are projectors onto eigenspaces of $G$ and $G'$ with eigenvalues from $S_{1}$ and $S_{2}$.
  The sets $S_{1}$ and $S_{2}$ must be separated by an annulus or infinite strip of width $\Delta$ in the complex plane.
  We set $S_{1} = [1, \Gamma]$, $S_{2} = [-\delta', \delta']$ and $\Delta = 1 - \delta'$. We denote by $\ket{\psi} = \ket{\psi(t)}$ and $\ket{\psi'}$ the (normalized) ground states of $G$ and $G'$.
  Then $E = \idm - \ketbra\psi\psi$, $F = \ketbra{\psi'}{\psi'}$ and
  \begin{align}
    1 - \abs{\braket{\psi}{\psi'}} 
    &=
      \opnormd{\ketbra{\psi'}{\psi'}}
      - \abs{\braket{\psi}{\psi'}} \opnormd{\ketbra{\psi}{\psi'}}
      \nonumber
    \\
    &\le
      \opnormd{\ketbra{\psi'}{\psi'} - \ket\psi \braket{\psi}{\psi'} \bra{\psi'}}
      \nonumber
    \\ 
    &=
      \opnorm{EF}
      \le
      \frac{\delta'}{1 - \delta'}
      \le 2\delta'
  \end{align}
  where the very last inequality holds for $\delta' \le \frac12$.
  The change in the ground state is at most $1 - \abs{\braket{\psi}{\psi'}} \le \delta' / (1 - \delta') \le 1$ (using $\delta' \le \frac12$). This implies (Lemma~\ref{lem:local-te-cert-fid-tracedist})
  \begin{align}
  \norm{\psi - \psi'}_{1} \le
  2 \sqrt{\frac{2\delta'}{1 - \delta'}} \le 4 \sqrt{\delta'}
  \label{eq:approx-time-evo-parent-H-gs-change}
  \end{align}
  and completes the proof.  
\end{proof}

\begin{rem}
  The certificates provided by \cref{thm:time-evo-approx-parent-h} and \cref{lem:timeevo-approx-parent-h-cert} differ in that the former certifies the fidelity with the time-evolved state $\ket\psi$ while the latter certifies the fidelity with its approximation $\ket{\psi'}$.
  The value of the infidelity upper bound provided by \cref{lem:timeevo-approx-parent-h-cert} is slightly larger than that provided by \cref{thm:time-evo-approx-parent-h}, but in the limit $\mci \to 0$ both results have the same scaling including all constants.
\end{rem}

\section{Efficient representation of time evolution%
  \label{sec:timeevo-repr}}

In this section, we construct a unitary circuit which approximates the unitary evolution $U_{t}$ induced by a local Hamiltonian $H(t)$ on $n$ quantum systems; the circuit approximates $U_{t}$ up to operator norm distance $\epsilon$.
For times poly-logarithmic in $n$, the circuit is seen to admit an efficient \gls{peps} representation; hence, the circuit shows that $U_{t}$ can be approximated by an efficient \gls{peps}.

Note that the following line of argument also provides an efficient \gls{peps} representation of $U_{t}$. 
Time evolution under an arbitrary few-body Hamiltonian can be efficiently simulated with a unitary quantum circuit and the Trotter decomposition \parencite[Chapter~4.7.2]{Nielsen2007}.
This unitary circuit is efficiently encoded as a \gls{mbqc}.
In turn, a \gls{peps} of the smallest non-trivial bond dimension two is sufficient to encode an arbitrary \gls{mbqc} efficiently \parencite{Schuch2007}. 
The \gls{peps} representation from this construction is efficient but it is supported on a larger lattice than the original Hamiltonian:
For example, the lattice grows as $\bigo{n t^{2}/\epsilon}$ if the the first-order Trotter decomposition is used (cf.~\cref{sec:sec:timeevo-repr-trotter}).
Application of the Trotter formula also leads to an efficient representation of $U_{t}$ as a \gls{tns} but the lattice of this construction grows in the same way \parencite{Huebener2010}.
Here, we construct an efficient \gls{peps} representation of $U_{t}$ which lives on the same lattice as the Hamiltonian and which has another advantageous property:
Computing the expectation value of a local observable in an arbitrary \gls{peps} is assumed to be impossible in polynomial time \parencite{Schuch2007} but the unitary circuit from which we construct our \gls{peps} representation always enables efficient computation of such local expectation values.
This property is shared e.g.\ with the class of so-called \glspl{bsgs}, a subclass of all \gls{peps}, where a state is also represented by a sequence of local unitary operations (albeit aranged differently; \cite{Banuls2008}).

The limitations of the first-order Trotter decomposition become apparent already in one spatial dimension as discussed in \cref{sec:sec:timeevo-repr-trotter}. 
\Cref{sec:timeevo-peps} defines \glspl{peps} on an arbitrary graph and determines an upper bound for the \gls{peps} bond dimension of a unitary circuit based on an argument used before for \glspl{mps} \parencite{Jozsa2006}.
\Cref{sec:timeevo-repr-arb-lattice} presents an efficient representation of $U_{t}$ for an arbitrary graph.
This representation is non-optimal in the sense that it evolves local observables into observables which seemingly act non-trivially on a region whose diameter grows polynomially with time.
Lieb--Robinson bounds already tell us that this diameter should grow only linearly with time (\cref{sec:timeevo-lr}).
An improved representation which fulfills this property is presented in \cref{sec:timeevo-repr-cubic-lattice} for a hypercubic lattice of spatial dimension $\eta \ge 1$. 

\subsection[Trotter decomposition]{Properties of the Trotter decomposition%
  \label{sec:sec:timeevo-repr-trotter}}

The Trotter decomposition is the key ingredient of many numerical methods for the computation of $U
_{t}$ with \glspl{mps} or \glspl{peps}.\footnote{%
  E.g.\ \textcite{Vidal2004,Murg2007,Verstraete2008} and references in \textcite{Schollwoeck2011}.
}
As discussed above, it also enables various efficient representations of $U_{t}$.
The following \namecref{lem:trotter-scaling} presents the well-known first-order Trotter decomposition:

\begin{lem}[Trotter decomposition in 1D]
  \label{lem:trotter-scaling}
  
  Let $H$ a time-independent\footnote{%
    The time-dependent case is discussed e.g.\ by \textcite{Poulin2011a}.
  }
  nearest neighbour Hamiltonian on a linear chain of $n$ spins, $H = \sum_{j=1}^{n-1} h_{j,j+1}$.
  Let the operator norm of the local terms be uniformly bounded, i.e.\ $\opnorm{h_{j,j+1}} \le J$ ($j \in \{1 \dots n-1\}$). 
  Take $H_{1}$ and $H_{2}$ to be the sum of the terms with even and odd $j$, respectively:
  Set $H_{1} = \sum_{j=1}^{\floor{(n-1)/2}} h_{2j,2j+1}$ and $H_{2} = \sum_{j=1}^{\floor{n/2}} h_{2j-1,2j}$.
  The time evolution induced by $H$ is given by $U_{t} = \ee^{-\ii H t}$ and its Trotter approximation is given by $U^\text{(T)}_{t} = \parenc{ \ee^{-\ii H_{1} \tau}\ee^{-\ii H_{2} \tau} }^{L}$ where $L$ is a positive integer and $\tau = t / L$. 
  The approximation error is at most $\epsilon$, i.e.
  \begin{align}
    \opnorm{ U_{t} - U^\text{(T)}_{t} } \le \epsilon
  \end{align}
  if $L$ is at least $L \ge \tilde c t^{2} n / (2\epsilon)$ where $\tilde c > 0$ is some constant which depends only on $J$. 
\end{lem}

\begin{proof}
  For any division of $H$ into $H = H_{1} + H_{2}$ and any $\tau \ge 0$, the following inequality holds:\footnote{%
  This is Eq.~(A.15a) of \textcite{DeRaedt1987}. See also \textcite{Suzuki1985}. 
}
  \begin{align}
  \opnormc{ \ee^{-\ii H \tau} - \ee^{-\ii H_{1} \tau}\ee^{-\ii H_{2} \tau} }
  \le
  \frac{\tau^{2}}2 \opnormc{[H_{1}, H_{2}]}.
  \end{align}
  Using the triangle inequality (as in \cref{lem:unit-seq-triangle-unit-inv}) and $\tau = t / L$, we obtain
  \begin{align}
  \opnormc{ \ee^{-\ii H t} - \parenc{ \ee^{-\ii H_{1} \tau}\ee^{-\ii H_{2} \tau} }^{L} }
  \le
  \frac{L \tau^{2}}2 \opnormc{[H_{1}, H_{2}]}.
  \end{align}
  It is simple to show that $\opnorm{[H_{1}, H_{2}]} \le \tilde c n$ holds for some constant $\tilde c > 0$ which depends only on $J$.
  This provides
  \begin{align}
    \opnorm{ U_{t} - U^\text{(T)}_{t} } \le \frac{\tilde c t^{2} n}{2L}
    \label{eq:trotter-scaling}
  \end{align}
  which completes the proof.
\end{proof}

\begin{figure}[t]
  \centering
  \includegraphics[scale=.6]{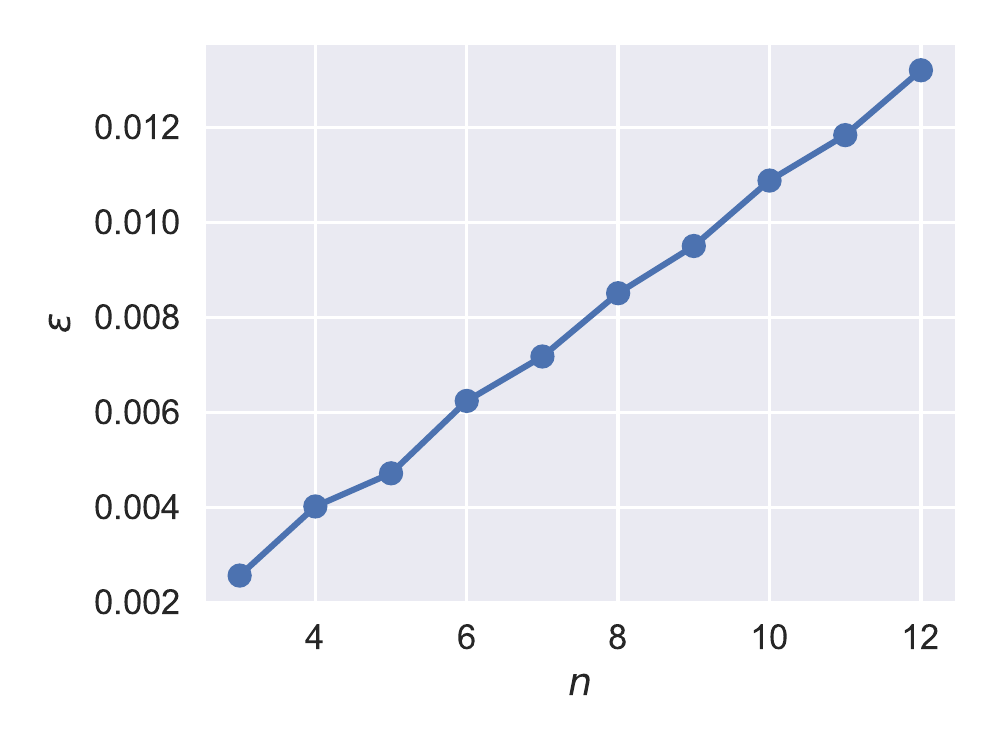}
  \caption{
    Operator norm error $\epsilon = \opnorm{ U_{t} - U^\text{(T)}_{t} }$ of the first order Trotter decomposition as a function of the number of spins $n$.
    The figure shows data for $L = 200$ Trotter steps and the 1D nearest-neighbour isotropic Heisenberg Hamiltonian at $t = \frac{11}{9J}$ where $J = \norm{h_{i,i+1}}$ is the operator norm of a coupling term.}
  \label{fig:trotter-first-order-error}
\end{figure}

\Cref{fig:trotter-first-order-error} shows the approximation error $\epsilon = \opnorm{ U_{t} - U^\text{(T)}_{t} }$ of a particular Hamiltonian as function of $n$ at fixed $t$ and $L$.
The approximation error appears to grow linearly with $n$ and this suggests that the bound \eqref{eq:trotter-scaling} is optimal in $n$ up to constants;
in this case, the scaling $L \ge \tilde c t^{2}n / \epsilon$ is optimal in $n$ up to constants as well.

\Cref{lem:trotter-scaling} provides an approximate decomposition of $U_{t}$ into $\bigo{n^{2} t^{2}/\epsilon}$ two-body unitaries and it has been recognized before that this constitues an approximate, efficient decomposition of $U_{t}$ by a tensor network on a two-dimensional lattice with $\bigo{n^{2}t^{2}/\epsilon}$ sites \parencite{Huebener2010}.
However, the lattice of the Hamiltonian is only one-dimensional.
The bond dimension of a one-dimensional \gls{mpo} representation of the circuit $U^\text{(T)}_{t}$ can grow exponentially with $n$.\footnote{%
  See \textcite{Osborne2006}. This can be seen by applying the counting argument by \textcite{Jozsa2006}, which is also stated below in \cref{lem:peps-bond-dim-upper-bound} for a more general \gls{peps}. 
}
It has been shown that $U_{t}$ indeed admits a smaller bond dimension \parencite{Osborne2006} but this is not visible from the circuit $U^\text{(T)}_{t}$ provided by the Trotter decomposition and needs additional arguments based on Lieb--Robinson bounds. 
Since the first-order Trotter decomposition does not provide an efficient \gls{mpo} representation of $U_{t}$ with $H$ on a one-dimensional lattice, it does not provide an efficient \gls{peps} representation on the same lattice as the Hamiltonian in higher dimensions either. 

Another important property of representations of the time evolution $U_{t}$ concerns the growth of the region on which a time-evolved, initially local observable appears to act non-trivially.
If an initially local observable $A$ is evolved with the Trotter decomposition $U^\text{(T)}_{t}$ into $(U^\text{(T)}_{t})^{\adjm} A U^\text{(T)}_{t}$, it appears to act non-trivially on a region of diameter $\bigo{L} = \bigo{n t^{2}/\epsilon}$.
In the following \cref{sec:timeevo-repr-arb-lattice,sec:timeevo-repr-cubic-lattice}, we construct circuits under which this diameter grows only poly-logarithmically with $n/\epsilon$. 
This is an improvement over the Trotter circuit but it does not reach the ideal case from \cref{cor:timeevo-quasilocality-D-scaling} (no growth with $n$).

\subsection{Projected entangled pair states (PEPSs)
  \label{sec:timeevo-peps}}

In the following, we define the \gls{peps} representation of a quantum state of $n < \infty$ quantum systems.
In order to introduce the \gls{peps} representation, we identify the pure quantum state $\ket\psi$ on $n$ systems with a tensor $t$ with $n$ indices. 
Let $\Lambda = \{1 \dots n\}$ be the set of all systems.\footnote{%
  The systems need not be in a linear chain but we assign the names $1$, \dots, $n$ to the sites of the system in an arbitrary order.
}
Let $d(x)$ denote the dimension of system $x \in \Lambda$.
Let $\ket{\phi^{(x)}_{i}}$ ($i \in \{1 \dots d(x)\}$) denote an orthonormal basis of system $x$.
The components of a pure state $\ket\psi$ on the $n$ systems are given by
\begin{align}
  t_{i_{1}\dots i_{n}} &= \braket{\phi^{(1)}_{i_{1}} \dots \phi^{(n)}_{i_{n}}}{\psi},
  &
  i_{x} &\in \{1 \dots d(x)\}.
  \label{eq:state-as-tensor}
\end{align}
The last equation shows that the pure state on $n$ systems corresponds to a tensor $t$ with $n$ indices of shape $d(1) \times \dots \times d(n)$.
A \gls{peps} representation of $\ket\psi$ or $t$ is defined in terms of a graph $(\Lambda, E)$ whose vertices correspond to sites $x \in \Lambda$ (\cref{fig:lattice-and-peps-tn} left).
Whenever we combine \gls{peps} representations and the Lieb--Robinson bounds from \cref{sec:timeevo-lr}, it is mandatory that the metric $d(x, y)$ on $\Lambda$ is the graph metric of the graph $(\Lambda, E)$ which defines the \gls{peps} representation.
The graph $(\Lambda, E)$ is assumed to be connected and simple, i.e.\ each edge $e \in E$ connects exactly two distinct sites.
The set of neighbours of $x \in \Lambda$ is given by $N(x) = \{y \in \Lambda\colon \{x, y\} \in E\}$ and the number of neighbours (degree) is given by $z_{x} = \abs{N(x)}$.
We denote the edges involving $x \in \Lambda$ in some arbitrary, fixed order by $\{n^{(x)}_{1} \dots n^{(x)}_{z_{x}}\}$;
i.e.\ $n^{(x)}_{k} = \{x, y\} \in E$ for one $y \in N(x)$. 
For each edge $e \in E$, choose a positive integer $D(e)$, called the bond dimension.
The maximal local and bond dimension are denoted by $d = \max_{x \in \Lambda} d(x)$ and $D = \max_{e\in E} D(e)$.
For $x \in \Lambda$, let $G_{x}$ a tensor of size $d(x) \times D(n^{(x)}_{1}) \times \dots \times D(n^{(x)}_{z_{x}})$.
Let $\{e_{1} \dots e_{\abs E}\} = E$ an enumeration of all the edges.
A \gls{peps} representation of the tensor $t$ is given by (\cref{fig:lattice-and-peps-tn} middle)
\begin{align}
  t_{i_{1}\dots i_{n}} =
  \sum_{b(e_{1})=1}^{D(e_{1})} \dots \sum_{b(e_{\abs E})=1}^{D(e_{\abs E})}
  \prod_{x=1}^{n}
  G_{x}\sqbd{i_{x}, b\parenc{n^{(x)}_{1}}, \dots, b\parenc{n^{(x)}_{z_{x}}}}
  \label{eq:peps-repr}
\end{align}%
\begin{figure}[t]
  \centering
  \isvgc{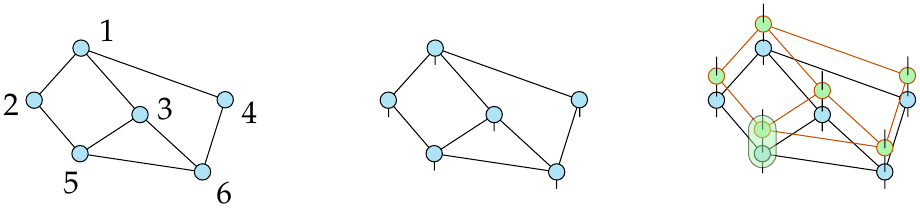}
  \caption{Left:
    Graph $(\Lambda, E)$ where the vertices correspond to lattice sites $x \in \Lambda$.
    Lattice sites have been named $\{1\dots6\} = \Lambda$.
    Middle: Graphical representation of the tensor network which constitutes a \gls{peps} representation of a quantum state on the lattice.
    Circles correspond to tensors, lines correspond to indices and lines which connect two circles indicate which indices are contracted (\cref{eq:peps-repr}).
    Right: The shaded box exemplifies how a \gls{pepo} representation of the operator product is obtained (\cref{eq:peps-product-local-tensors}). 
  }
  \label{fig:lattice-and-peps-tn}%
\end{figure}%
A \gls{peps} representation of a pure quantum state is given by the combination of \cref{eq:state-as-tensor,eq:peps-repr}. 
Any tensor or quantum state can be represented as \gls{peps} if the bond dimensions $D(e)$ are made sufficiently large (cf.\ \cref{lem:peps-bond-dim-upper-bound} below).

\glsreset{pepo}

A \gls{pepo} representation of a linear operator $G$ on $n$ quantum systems is given by a \gls{peps} representation of the following tensor:
\begin{align}
  t_{(i_{1},j_{1}) \dots (i_{n},j_{n})} = \bra{\phi^{(1)}_{i_{1}} \dots \phi^{(n)}_{i_{n}}} G \ket{\phi^{(1)}_{j_{1}} \dots \phi^{(n)}_{j_{n}}}
\end{align}
Here, $t$ is considered as tensor with $n$ indices and size $[d(1)]^{2} \times \dots \times [d(n)]^{2}$.
Suppose that two linear operators $G$ and $H$ have \gls{pepo} representations given by tensors $G_{x}$ and $H_{x}$ with bond dimensions $D_{G}(e)$ and $D_{H}(e)$.
The following formula provides the tensors of a \gls{pepo} representation of the operator product $F = GH$ (\cref{fig:lattice-and-peps-tn} right):
\begin{align}
  &
    F_{x}\sqbd{i_{x}, k_{x}, b\parenc{n^{(x)}_{1}}, \dots, b\parenc{n^{(x)}_{z_{x}}}}
    \nonumber
  \\
  &=
    \sum_{j_{x}=1}^{d(x)}
    G_{x}\sqbd{i_{x}, j_{x}, c\parenc{n^{(x)}_{1}}, \dots, c\parenc{n^{(x)}_{z_{x}}}}
    H_{x}\sqbd{j_{x}, k_{x}, d\parenc{n^{(x)}_{1}}, \dots, d\parenc{n^{(x)}_{z_{x}}}}
    \label{eq:peps-product-local-tensors}
\end{align}
where $b\parenc{n^{(x)}_{k}} = \parenc{c\parenc{n^{(x)}_{k}}, d\parenc{n^{(x)}_{k}}}$ ($k \in \{1 \dots z_{x}\}$).
\Cref{eq:peps-product-local-tensors} proves the following \namecref{lem:peps-product-bond-dimension}:

\begin{lem}
  \label{lem:peps-product-bond-dimension}
  Let $G$ and $H$ be operators with \gls{peps} bond dimensions $D_{G}(e)$ and $D_{H}(e)$.
  The operator product $F = GH$ admits a $\gls{peps}$ representation with bond dimension $D_{F}(e) = D_{G}(e) D_{H}(e)$ ($e \in E$). 
\end{lem}

The next \namecref{lem:peps-bond-dim-upper-bound} gives an explicit upper bound on the bond dimension of the \gls{peps} representation of an arbitrary tensor:

\begin{lem}
  \label{lem:peps-bond-dim-upper-bound}
  Let $t$ a tensor with $n$ indices and size $d(1) \times \dots \times d(n)$.
  Then $t$ admits a \gls{peps} representation with maximal bond dimension $D \le d^{n}$ where $d = \max_{x \in \Lambda}d(x)$.
\end{lem}

\begin{rem}
  \label{rem:peps-bond-dim-upper-bound}
  When applying \cref{lem:peps-bond-dim-upper-bound} to an operator which acts non-trivially on a region $Y \subset \Lambda$ remember that this region $Y$ must be connected in terms of the \gls{peps} graph $(\Lambda, E)$.
\end{rem}

\glsreset{dfs}

\begin{proof}
  Suppose that the connected, simple graph $(\Lambda, E)$ is such that it admits a permutation $(x_{1}, \dots, x_{n})$ of all the vertices such that $\{x_{k-1}, x_{k}\} \in E$ is a valid edge ($k \in \{2 \dots n\}$; such a permutation is called a \emph{Hamiltonian path}).
  In this case, an \gls{mps}/\gls{tt} representation of the suitably permuted tensor provides a valid \gls{peps} representation with bond dimension $D \le d^{\floor{n/2}} < d^{n}$ \parencite[e.g.]{Schollwoeck2011}.
  However, the graph may not admit such a permutation.\footnote{Example: A central vertex connected to three surrounding vertices.}
  In this case, we perform a \gls{dfs} on the graph to obtain a tree graph with the same vertices and a subset of the edges of the original graph (we can start the \gls{dfs} on any vertex).
  Walking through the resulting tree graph in the \gls{dfs} order visits each vertex at least once and each edge at most twice.\footnote{%
    Tarry's algorithm returns a bidirectional double tracing, i.e.\ a walk over the graph which visits each edge exactly twice \parencite[Sec.~4.2.4]{Gross2014}.
    Omitting visits to already-visited vertices in this walk represents a depth-first search \parencite[Sec.~2.1.2]{Gross2014}.
  }
  The tensor with indices permuted according to their first visit in the \gls{dfs} order%
  \footnote{%
    It would be equally permissible to use the second or a later visit in the \gls{dfs} order.
  }
  can be represented as \gls{mps}/\gls{tt} of bond dimension $d^{\floor{n/2}}$.
  Because each edge is visited at most twice, the resulting \gls{mps} can be converted to a \gls{peps} with bond dimension $D \le (d^{\floor{n/2}})^{2} \le d^{n}$.
\end{proof}

The bond dimension of a unitary circuit can be bounded with \cref{lem:peps-product-bond-dimension,lem:peps-bond-dim-upper-bound} as follows:

\begin{lem}
  \label{lem:unit-circuit-eff-peps}
  Let $U = U_{1} \dots U_{G}$ a unitary circuit composed of $G$ gates where each gate $U_{j}$ acts non-trivially on at most $K$ connected sites ($j \in \{1\dots G\}$).
  Let at most $L$ gates act on any site of the system.
  The unitary $U$ admits an exact \gls{pepo} representation of bond dimension $D \le d^{2KL}$ where $d = \max_{x\in \Lambda}d(x)$ is the maximum local dimension. 
\end{lem}

\begin{proof}
  The statement is proven by repeating a simple counting argument which has been used before by \textcite{Jozsa2006} for one-dimensional \gls{mps}.\footnote{%
    The argument could be improved by counting how often each edge is used instead of counting how often each site is used; cf.\ \textcite{Holzaepfel2014}.
  }
  As the operator $U_{j}$ acts non-trivially on at most $K$ connected sites, it admits a \gls{pepo} representation with bond dimension $D' = (d^{2})^{K}$ (\cref{lem:peps-bond-dim-upper-bound}).
  At each edge, the bond dimension of $U$ is at most the product of the bond dimensions of the operators $U_{1}$, \dots, $U_{G}$ (\cref{lem:peps-product-bond-dimension}): $D_{U}(e) \le \prod_{j=1}^{G} D_{U_{j}}(e)$, $e \in E$.
  We have $D_{U_{j}}(e) \le D'$ for all edges and $D_{U_{j}}(e) = 1$ if the edge $e$ involves a site on which $U_{j}$ acts as the identity. 
  At most $L$ of the $G$ operators $U_{1}$, \dots, $U_{G}$ act non-trivially on an arbitrary site $j$ and this bounds $U$'s bond dimension to $D_{U}(e) \le (D')^{L} = d^{2KL}$ for all edges.
\end{proof}

\subsection[Arbitrary lattice]{Efficient representation of time evolution: Arbitrary lattice%
  \label{sec:timeevo-repr-arb-lattice}}

Suppose that a local Hamiltonian $H(t)$ is perturbed by a spatially local and possibly time-dependent perturbation $A(t)$.
The following \namecref{lem:osborne-hamilton-perturbation} states that there is a spatially local unitary $V'$ such that $\opnorm{V' U^{H - A}_{ts} - U^{H}_{ts}}$ is small; the \namecref{lem:osborne-hamilton-perturbation} has been proven for one-dimensional systems by \textcite{Osborne2006}.
His proof also works for higher-dimensional systems if combined with \cref{thm:timeevo-quasilocality} (proven by \cite{Barthel2012}). 
We pretend to extend the existing proof by accounting for time-dependent Hamiltonians explicitly.

\begin{lem}
  \label{lem:osborne-hamilton-perturbation}
  
  Let $a$, $\mcz$ and $J$ be finite and $t \in \R$.
  Let $Y \subset R \subset \Lambda$ and let $A(s)$ act on $Y$.
  Let $A(s)$ be continuous except for finitely many discontuinuities in any finite interval.
  Choose $q \in (0, 1)$ and set $D = (1-q) d_{a}$ where $d_{a} = d(Y, \Lambda \setminus R) / a$. 
  Let $\ceil{d_{a}}$ large enough (\cref{lem:timeevo-scale-distance}). 
  Let $V'_{s}(t)$ on $\bar R$ be the solution of $\partial_{s} V'_{s}(t) = \ii L'_{t}(s) V'_{s}(t)$ where $L'_{t}(s) = \tau^{H_{\bar R}}_{ts} (A(s))$ and $V'_{t}(t) = \idm$ ($s \in \R$).
  Then
  \begin{align}
    \opnorm{ V'_{s}(t)U^{H-A}_{ts} - U^{H}_{ts} }
    \le
    \frac{2M \alpha_{q}}{v \mcz} \abs A \exp(v\abs{t - s} - D)
  \end{align}
  where $\abs A = \max_{r\in [s, t]} \opnorm{A(r)}$.
  The Lieb--Robinson velocity is given by $v = J \mcz \exp(1)$ and $\alpha_{q} \in (1 / \ee^{1-q}, 1]$.
  Specifically, $\alpha_{q} = \exp(-(1-q) (\ceil{d_{a}} - d_{a}))$.
\end{lem}

\begin{proof}
  Let $V_{s}(t) = U^{H}_{ts} U^{H-A}_{st}$.\footnote{%
    Alternatively, one can obtain an approximation of the form $U^{H}_{ts} \approx U^{H-A}_{ts} W'_{t}(s)$ where $W'_{t}(s)$ is the solution of $\partial_{t} W'_{t}(s) = -\ii W'_{t}(s) \tau^{H_{\bar R}}_{st}(A_{t})$, $W'_{s}(s) = \idm$.
    $W'_{t}(s)$ is an approximation of $W_{t}(s) = U^{H-A}_{st} U^{H}_{ts}$. 
    This approach is a bit more similar to the original proof by \textcite{Osborne2006}.
  }
  Due to unitary invariance of the operator norm, we have
  \begin{align*}
    \opnorm{V'_{s}(t) U^{H-A}_{ts} - U_{ts}}
    = \opnorm{(V'_{s}(t) U^{H-A}_{ts} - U_{ts}) U^{H-A}_{st}}
    = \opnorm{V'_{s}(t) - V_{s}(t)}.
  \end{align*}
  For fixed $t \in \R$, $V_{s}(t)$ satisfies the differential equation
  \begin{align}
    \partial_{s} V_{s}(t) = \ii U^{H}_{ts}(H(s) - H(s) + A(s)) U^{H-A}_{st}
    = \ii L_{t}(s) V_{s}(t)
  \end{align}
  where $L_{t}(s) = U^{H}_{ts} A(s) U^{H}_{st} = \tau^{H}_{ts}(A(s))$ and $V_{t}(t) = \idm$.
  The Trotter decomposition of $V$ is given by%
  \footnote{
    See e.g.\ Theorem~1.1 and~1.2 by \textcite{Dollard1979a} or Theorem~3.1 and~4.3 by \textcite{Dollard1979}.
  }
  \begin{align}
    V_{t}(s) = 
    \lim_{m\to\infty}
    \prod_{j=1}^{m}
    \ee^{-\ii L_{t}(s + j\delta_{m}) \delta_{m}}
  \end{align}
  where $\delta_{m} = (t-s) / m$.
  The operator norm is unitarily invariant, therefore the triangle inequality implies $\norm{U_{1} U_{2} - V_{1} V_{2}} \le \norm{U_{1} - V_{1}} + \norm{U_{2} - V_{2}}$ (\cref{lem:unit-seq-triangle-unit-inv}).
  We obtain
  \begin{subequations}
  \begin{align}
    \opnorm{V_{s}(t) - V'_{s}(t)}
    &\le
      \lim_{m \to \infty} \sum_{j=1}^{m} 
      \opnorma{\ee^{-\ii L_{t}(s + j\delta_{m}) \delta_{m}} -   \ee^{-\ii L'_{t}(s + j\delta_{m}) \delta_{m}}}
    \\
    &\le
      \lim_{m \to \infty} \sum_{j=1}^{m}
      \abs{\delta_{m}} \opnorma{L_{t}(s + j \delta_{m}) - L'_{t}(s + j \delta_{m})}
    \\
    &=
      \int_{s}^{t} \opnorma{L_{t}(r) - L'_{t}(r)} \vd r.
      \label{eq:osborne-hamilton-perturbation-integral}%
  \end{align}%
  \end{subequations}
  For all $r, r' \in [s, t]$, \cref{cor:timeevo-quasilocality} provides the bound
  \begin{align}
    \opnorma{\tau^{H}_{tr}(A(r')) - \tau^{H_{\bar R}}_{tr}(A(r'))}
    \le    
    \frac{2M \alpha_{q}}{\mcz} \opnorm{A(r')} \exp(v\abs{t-r} - D). 
  \end{align}
  Inserting $r' = r$ provides a bound on $\opnorma{L_{t}(r) - L'_{t}(r)}$; inserting this bound into \eqref{eq:osborne-hamilton-perturbation-integral} completes the proof.
\end{proof}

In the following \namecref{lem:time-evo-eff-repr}, we decompose the global evolution $U_{ts}$ into a sequence of local unitaries by removing all local terms of the Hamiltonian which involve site $n$, then removing those which involve site $n-1$ and so on.
Here, the order of the sites does not matter and the geometry of the lattice enters only via the constants introduced before.
However, the subsequent \cref{thm:time-evo-eff-repr} shows that ordering the sites of the system in a certain way improves the properties of the resulting unitary circuit.

\begin{lem}
  \label{lem:time-evo-eff-repr}
  
  Let $H_{j} = \sum_{Z \subset \Lambda_{j}} h_{Z}$ denote the sum of all terms which act on the first $j$ sites $\Lambda_{j} = \{1 \dots j\}$.
  Denote by $Y_{j} \subset \Lambda_{j}$ the set of sites on which $F_{j} = H_{j} - H_{j-1}$ acts non-trivially.
  Choose $q \in (0, 1)$.
  Let $r$ be such that $\ceil{r} > 2\kappa + 1$ and $\ceil{r} \ge \frac{2\kappa}{q} \ln(\frac{\kappa}{q})$.
  Let $R = (1-q)r$ satisfy
  \begin{align}
    R \ge v\abs{t-s} + \ln\parena{\frac n\epsilon} + c_{2}
  \end{align}
  where $c_{2} = \ln\parena{\frac{M}{\mcz \exp(1)}} + 2(1-q)$.
  Set $R_{j} = B^{o}_{da}(Y_{j}) \cap \Lambda_{j}$ where $d = r - 2$.
  Let $\bar R_{j}$ be the extension of $R_{j}$ in terms of the Hamiltonian $H_{j}$, i.e.\ $\bar R_{j} \subset \Lambda_{j}$.
  Then $\bar R_{j} \subset B^{o}_{ra}(\{j\}) \cap \Lambda_{j}$. 
  Let $V'_{js}(t)$ on $\bar R_{j}$ be the solution of $\partial_{s} V'_{js}(t) = \ii \tau^{G_{j}}_{ts}(F_{j}(s)) V'_{js}(t)$ where $G_{j} = H_{\bar R_{j}}$ and $V'_{jt}(t) = \idm$. 
  Then
  \begin{align}
    \opnorm{U^{H}_{ts} - V'_{n} \dots V'_{2} V'_{1}}
    \le
    \epsilon
  \end{align}
  holds where $V'_{j} = V'_{js}(t)$.
\end{lem}

\begin{proof}
  There are at most $\mcz$ non-zero local terms $h_{Z}$ with $j \in Z$.
  As a consequence, $\opnorm{F_{j}(s)} \le \mcz J / 2$ holds.
  In addition, $Y_{j} \subset B^{c}_{a}(\{j\})$ holds and implies $\bar R_{j} \subset B^{c}_{a}(R_{j}) \cap \Lambda_{j} \subset B^{c}_{a}(B^{o}_{da}(B^{c}_{a}(\{j\}))) \cap \Lambda_{j} \subset B^{o}_{ra}(\{j\}) \cap \Lambda_{j}$ (\cref{lem:finite-metric-space}). 
  The definitions imply that $d(Y_{j}, \Lambda_{j} \setminus R_{j}) / a \ge d$  (\cref{lem:finite-metric-space}).\footnote{%
    Note that we restrict to the sublattice $\Lambda_{j}$.
  }
  Set $D = (1-q) d$.
  Note that $D = R - 2(1-q)$.
  Therefore, \cref{lem:osborne-hamilton-perturbation} implies that
  \begin{align}
    \opnorm{ U^{H_{j}}_{ts} - V'_{j} U^{H_{j}-F_{j}}_{ts} } \le
    \frac{J M}{v} \exp(v \abs{t-s} - D)
    \le
    \frac{J M}{v} \frac{\epsilon}{n} \exp(2(1-q) -c_{2})
    =
    \frac{\epsilon}{n}
  \end{align}
  holds for all $j \in \{1\dots n\}$. Note that we have
  \begin{align}
    U^{H}_{ts} - V'_{n} \dots V'_{2} V'_{1} \quad=\quad
    \sum_{j=1}^{n} V'_{n} \dots V'_{j+1} U^{H_{j}}_{st} - V'_{n} \dots V'_{j} U^{H_{j-1}}_{ts}
  \end{align}
  where $H = H_{n}$ and $U^{H_{0}}_{ts} = \idm$. The triangle inequality and unitary invariance of the operator norm imply
  \begin{align}
    \opnorm{U^{H}_{ts} - V'_{n} \dots V'_{2} V'_{1}}
    \le
    \sum_{j=1}^{n} \opnorm{ U^{H_{j}}_{ts} - V'_{j} U^{H_{j-1}}_{ts} }
    \le
    \epsilon
  \end{align}
  where we have used $H_{j-1} = H_{j} - F_{j}$.
  This completes the proof of the Lemma.
\end{proof}

\begin{cor}
  \label{cor:time-evo-eff-peps}
  Let $d$ be the graph metric of the \gls{peps} graph and let $t$ be poly-logarithmic in $n$. 
  The operator $V = V'_{n} \dots V'_{1}$ from \cref{lem:time-evo-eff-repr} provides an efficient, approximate \gls{peps} representation of the time evolution $U^{H}_{ts}$ because it admits a \gls{peps} representation with bond dimension $D = \poly(R)$.
  
  Specifically, the bond dimension is $D = d^{2n_{r}^{2}}$ where $n_{r} = \max_{j \in \Lambda} \abs{B^{o}_{ar}(\{j\})}$ is the maximal number of sites in a ball of radius $ar$. 
\end{cor}

\begin{proof}
  All open balls $B^{o}_{k}(Z)$ ($k \ge 0$, $Z \subset \Lambda$) are connected in terms of the \gls{peps} graph because $d$ is the graph metric of that graph. 
  The unitary $V'_{j}$ acts as the identity outside the connected set $B^{o}_{ar}(\{j\})$ which contains at most $n_{r} = \max_{j \in \Lambda} \abs{B^{o}_{ar}(\{j\})} = \poly(R)$ sites.
  At most $\abs{B^{o}_{ar}(\{j\})} \le n_{r}$ of the $n$ operators $V'_{1}$, \dots, $V'_{n}$ act non-trivially on a given, arbitrary site $j$.
  Applying \cref{lem:unit-circuit-eff-peps} with $K = L = n_{r}$ completes the proof. 
\end{proof}

\Cref{lem:time-evo-eff-repr} provides an efficient, approximate representation of the time evolution $U^{H}_{ts}$.
However, this representation may not be particularly useful:
Consider a one-dimensional setting where $V'_{j}$ acts only on $\{j, j-1\}$ and $A$ is an observable which acts on site $n$.\footnote{%
  Indeed, the operators $V'_{j}$ would need to act on larger numbers of neighbouring sites to achieve a non-zero value of $R$ if the Hamiltonian contains any interactions. 
}
We want to compute the expectation value $\tr(\tau^{H}_{st}(A) \rho(s))$ where the initial state $\rho(s)$ is a product state.
The time-evolved observable is given by $\tau^{H}_{st}(A) = U^{H}_{st} A U^{H}_{ts}$.
We could obtain an approximation from $\tau^{H}_{st}(A) \approx (V'_{1})^{\adjm} \dots (V'_{n})^{\adjm} A V'_{n} \dots V'_{1}$, but the latter operator can act non-trivially on the full system.
The structure of the approximation does not convey the fact that operators propagate with the finite Lieb--Robinson velocity, as shown e.g.\ by \cref{thm:timeevo-quasilocality}.
The next \namecref{thm:time-evo-eff-repr} shows how the representation can be improved by reordering the sites of the system before applying the \namecref{lem:time-evo-eff-repr}. 

\begin{thm}
  \label{thm:time-evo-eff-repr}
  Choose $R > 0$, $q > 0$ and set $r = R / (1 - q)$. 
  Let $L = \max_{j \in \Lambda} \abs{B^{o}_{2ar}(\{j\})}$.
  There is an efficiently computable colouring function $C \colon \Lambda \to \{1 \dots L\}$ which has the property that $C(x) = C(y)$ implies $d(x, y) / a \ge 2r$.
  Suppose that the sites of the system are ordered such that there are integers $a_{k}$ with $k \in \{0 \dots L\}$, $a_{0} = 1$ and $a_{L} = n$ in terms of which the consecutive sites $\{a_{k-1}+1 \dots a_{k}\}$ have the same colour $k \in \{1 \dots L\}$.
  In this case, $V = V'_{n} \dots V'_{1}$ from \cref{lem:time-evo-eff-repr} can be expressed as $V = W_{L} \dots W_{1}$ where $W_{k} = V'_{a_{k}} \dots V'_{a_{k-1}+1}$.
  $V'_{j}$ and $V'_{j'}$ do not act non-trivially on the same site if $j$ and $j'$ have the same colour.
  At most $L$ of the $n$ operators $V'_{1}$, \dots, $V'_{n}$ act non-trivially on any given site $j \in \Lambda$. 
\end{thm}

\begin{proof}
  Consider a graph with sites given by $\Lambda$ and edges given by $E_{C} = \{\{x, y\} \colon x, y \in \Lambda, 0 < d(x, y) < 2ar\}$.
  The number of nearest neighbours (degree) of this graph is $L - 1$.
  A so-called greedy colouring of the graph $(\Lambda, E_{C})$, which can be computed in $\bigo{n L}$ time,\footnote{%
    A greedy colouring is obtained by picking a vertex which has not been assigned a colour and assigning the first colour which has not been assinged to any neighbour of the given vertex (neighbour in terms of $E_{C}$).
    See e.g.\ \textcite[Sec.~14.1, Heuristic~14.3, p.~363]{Bondy2008} or \textcite[Sec.~5.1.2, Fact~F13]{Gross2014}. 
  }
  has the property $d(x, y) < 2ar \Rightarrow C(x) \ne C(y)$.
  I.e.\ a greedy colouring already has the necessary property $C(x) = C(y) \Rightarrow d(x, y) \ge 2ar$. 
  Note that $B^{o}_{ar}(\{j\})$ and $B^{o}_{ar}(\{j'\})$ have an empty intersection if $d(j, j') \ge 2ar$ (\cref{lem:finite-metric-space}). Therefore, in this case, at most one of $V'_{j}$ and $V'_{j'}$ act non-trivially on any site. 
\end{proof}

\begin{rem}
  \label{rem:time-evo-eff-repr-thm}
  
  The operator $V$ from \cref{thm:time-evo-eff-repr} admits a \gls{peps} representation with the bond dimension mentioned in \cref{cor:time-evo-eff-peps}.

  Note that \cref{thm:time-evo-eff-repr} states that at most $L = \max_{j \in \Lambda} \abs{B^{o}_{2ar}(\{j\})}$ unitary operations act on a given site while we already know that this number is at most $n_{r} = \max_{j \in \Lambda} \abs{B^{o}_{ar}(\{j\})}$ (proof of \cref{cor:time-evo-eff-peps}). 
  This difference enables efficient computation of the colouring function which arranges the operations $V'_{j}$ into $L$ groups of non-overlapping operations. 
  
  Let $A$ act non-trivially only on site $j$.
  The advantage of \cref{thm:time-evo-eff-repr} over \cref{lem:time-evo-eff-repr} is that $V^{\adjm} A V$ now acts non-trivially at most on $n_{A} = \abs{B^{o}_{s}(\{j\})}$ sites where $s = 2arL = \poly(R)$, i.e.\ at most on $n_{A} = \poly(R)$ sites (use \cref{lem:op-support-growth} and \cref{lem:finite-metric-space}).
  This is an improvement over \cref{lem:time-evo-eff-repr} alone where $V^{\adjm} A V$ can (appear to) act non-trivially on the full system.
  The radius $s$ increases polynomially with $R$, i.e.\ polynomially with time.
  Below, we construct an improved representation where $s$ increases linearly with time (\cref{cor:timeevo-cubic-decomp-support}), which matches what is already known from Lieb--Robinson bounds (e.g.\ \cref{thm:timeevo-quasilocality}). 
\end{rem}

\begin{lem}
  \label{lem:op-support-growth}
  Let $A$ be an operator which acts non-trivially on $Y \subset \Lambda$.
  Let $V = V_{1} V_{2} \dots V_{b}$ where the $V_{k}$ ($k \in \{1 \dots b\}$) are unitary and $V_{k}$ acts non-trivially (at most) on $B^{o}_{r}(\{j_{k}\})$ with some $j_{k} \in \Lambda$ and $r > 0$.
  Let the $V_{k}$ commute pairwise, i.e.\ $[V_{k}, V_{l}] = 0$ for all $k, l \in \{1 \dots b\}$.
  Then $B = V^{\adjm} A V$ acts non-trivially at most on $B^{o}_{2r}(Y)$.
\end{lem}

\begin{proof}
  In the expression $B = V_{b}^{\adjm} \dots V_{1}^{\adjm} A V_{1} \dots V_{b}^{\adjm}$\,, all $V_{k}$ which commute with $A$ can be omitted (because a given $V_{k}$ commutes with all other $V_{l}$).
  In particular, all $V_{k}$ which do not act non-trivially on $Y$ can be omitted without changing~$B$. 
  Let $x \in \Lambda$ be a site on which $B$ acts non-trivially.
  If $x \in Y$ holds, $x \in B^{o}_{2r}(Y)$ holds as well and we are finished.
  In the following, let $x \not\in Y$. 
  Then, there is a $k \in \{1 \dots b\}$ such that $x \in B^{o}_{r}(\{j_{k}\})$.
  In addition, there is a $y \in B^{o}_{r}(\{j_{k}\}) \cap Y$ (otherwise, $V_{k}$ and $A$ commute and $V_{k}$ can be omitted from $B$).
  Note that $d(x, y) < 2r$ (the diameter of the given open ball).
  As a consequence, $x \in B^{o}_{2r}(Y)$ holds, which finishes the proof. 
\end{proof}

\subsection[Hypercubic lattice]{Efficient representation of time evolution: Hypercubic lattice%
  \label{sec:timeevo-repr-cubic-lattice}}

In this section, we construct a representation of time evolution under a local Hamiltonian which has a smaller bond dimension than the representation presented above.
In order to split the complete time evolution into independent parts in a more efficient way, we consider a cubic lattice $\Lambda$ of finite dimension $\eta$ with $L$ sites in each direction:
\begin{align}
  \label{eq:def-cubic-lattice}
  \Lambda = \{ (x_{1}, \dots x_{\eta}) \colon x_{i} \in [1:L], i \in [1:\eta] \}
\end{align}
Here, we used the notation $[1:L] = \{1, 2, \dots, L\}$ to denote a set of consecutive integers.
The total number of sites is $n = \abs\Lambda = L^{\eta}$. 
In this section, $\floorofa = \floor a$ denotes the interaction range rounded down. 

Formally, we use the Cartesian product $A \times B \times C = \{(a, b, c)\colon a \in A, b \in B, c \in C\}$ of sets $A$, $B$ and $C$.
Assuming suitable equivalence relations, the Cartesian product becomes associative, i.e.\ $A \times (B \times C) = (A \times B) \times C = A \times B \times C$. 
The Cartesian product has the basic property $(A \times B) \cap (C \times D) = (A \cap C) \times (B \cap D)$.
Powers of sets are given by the cartesian product, e.g.\ $[1:L]^{2} = [1:L] \times [1:L]$, and this allows us to write $\Lambda = [1:L]^{\eta} \subset \Z^{\eta}$ where $\Z$ is the set of all integers.

We assume a metric $d$ on $\Lambda$ which satisfies the property
\begin{align}
  \label{eq:metric-normalized}
  \abs{x_{i} - y_{i}} \le d(x, y)
  \quad \forall \;
  i \in [1:\eta]
\end{align}
For example, the metric induced by the vector-$p$ norm, $d(x, y) = [\sum_{i=1}^{\eta} \abs{x_{i} - y_{i}}^{p}]^{1/p}$ with $p \in [1, \infty]$, has this property.
Below, we partition the lattice into cubic sets defined as follows:
\begin{defn}
  \label{defn:metric-cube}

  Two points $x, y \in \Z^{\eta}$ define the cube $C(x, y) = \Lambda \cap \bigtimes_{i=1}^{\eta} [x_{i}:y_{i}]$.
  For a non-negative integer $r$, the enlarged cube is defined as $C_{r}(C(x, y)) = C(x - rv, y + rv)$ where $v = (1, 1, \dots, 1) \in \Z^{\eta}$. 
\end{defn}

We employ the following notation for Cartesian products: Let $c, d \in \Z$ and $x, y \in \Z^{\eta-1}$, then
\begin{align}
  \nonumber
  & [c:d]_{i} \times C(x, y) = \\
  & [x_{1}:y_{1}] \times \dots \times [x_{i-1}:y_{i-1}]
  \times [c:d] \times
    [x_{i}:y_{i}] \times \dots \times [x_{\eta-1}:y_{\eta-1}].
    \label{eq:cart-product-at-i}
\end{align}

We partition the full lattice $\Lambda$ into cubes $Q_{m}$ of size $\Omega$ and aim at splitting the full time evolution into independent evolutions on the cubes $Q_{m}$.
\Cref{fig:timeevo-cubic-partition} illustrates the partition $\Lambda = \bigdotcup_{m} Q_{m}$ and outlines the way forward.
The next \namecref{lem:lattice-partition-coupling-terms} identifies all local terms $h_{Z}$ which couple at least two cubes $Q_{m}$ and $Q_{m'}$:

\begin{figure}[t]
  \centering
  \isvgc{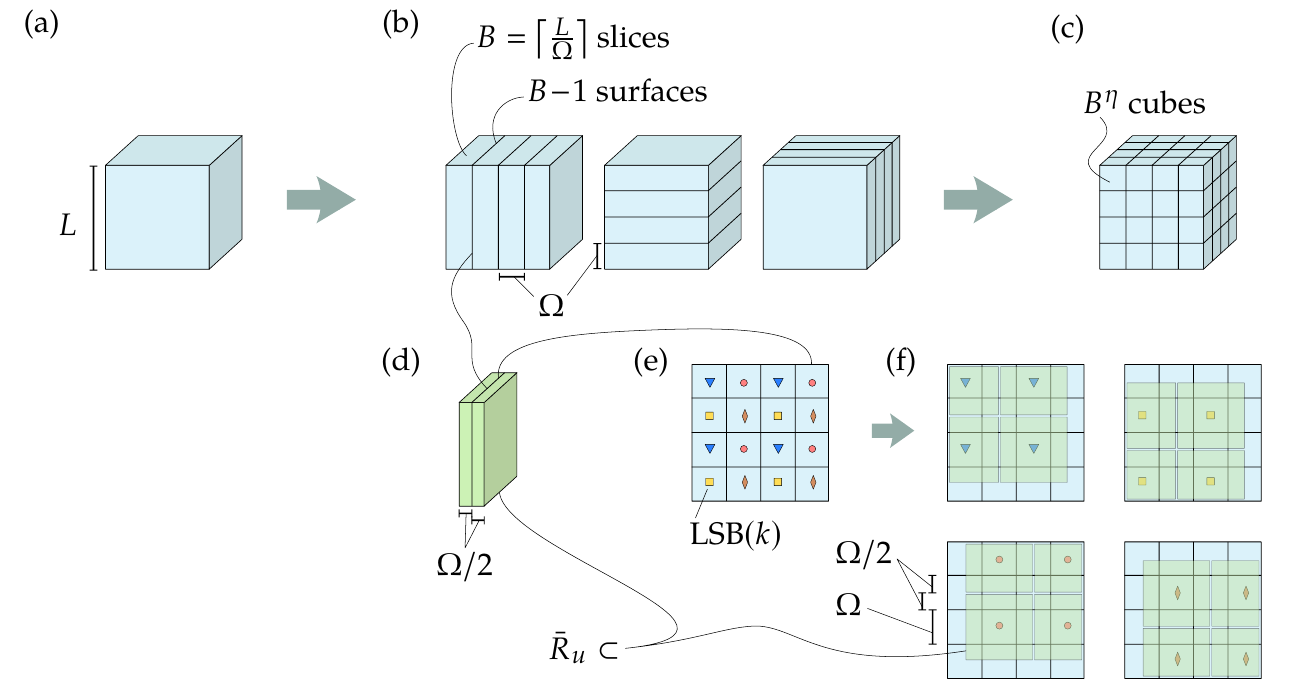}
  \caption{%
    Decomposition of an $\eta$-dimensional hypercube, illustrated for $\eta = 3$.
    (a)~The lattice $\Lambda$ lives inside a cube of edge length $L$.
    (b)~Along each direction, the cube is split into $B = \ceila{L/\Omega}$ slices of width $\Omega$ (\cref{lem:lattice-partition-coupling-terms}).
    (c)~As a result, the cube is split into $B^{\eta}$ smaller cubes $Q_{m}$ of edge length $\Omega$ (\cref{eq:lattice-partition-cubes}).
    (d)~Centered around each surface from~(b), there is a slice of width $\Omega$.
    (e)~Each surface from~(b) is split into $B^{\eta-1}$ surface segments $\tilde Q_{k}$ of edge length $\Omega$ (\cref{lem:coupling-term-partition}).
    The surface segments are divided into $2^{\eta-1}$ groups of non-neighbouring surfaces as specified by $\lsb(k)$ and indicated by the symbols (\cref{lem:def-time-evo-decomp-cubic}).
    (f)~The sets $\bar R_{u}$ do not overlap for surface segments with the same symbol (the same value of $\lsb(k)$) (\cref{lem:correction-support-overlap}). 
  }
  \label{fig:timeevo-cubic-partition}
\end{figure}

\begin{lem}
  \label{lem:lattice-partition-coupling-terms}

  Let $\Omega$ be a positive integer and $B = \ceil{L / \Omega}$.
  For $m \in [1:B]^{\eta}$, set
  \begin{align}
    \label{eq:lattice-partition-cubes}
    Q_{m} &= C(x_{m}, y_{m}),
    &
      x_{m} &= [\Omega (m_{i} - 1) + 1]_{i=1}^{\eta},
    &
      y_{m} &= [\Omega m_{i}]_{i=1}^{\eta}.
  \end{align}
  These cubes partition the lattice, $\Lambda = \bigdotcup_{m} Q_{m}$.
  For $i \in [1:\eta]$ and $j \in [1:B-1]$, set
  \begin{align}
    \label{eq:lattice-bipartition}
    A_{ij} &= [1:\Omega j]_{i} \times [1:L]^{\eta-1},
    &
    B_{ij} &= [\Omega j+1:L]_{i} \times [1:L]^{\eta-1}
  \end{align}
  and
  \begin{align}
    \label{eq:lattice-bipartition-terms}
    S_{ij} &= \{ Z \subset \Lambda \colon
             h_{Z} \ne 0, Z \cap A_{ij} \ne \emptyset, Z \cap B_{ij} \ne \emptyset \}.
  \end{align}
  The complete Hamiltonian is given by $H = H_{Q} + H_{S}$ where $H_{Q}$ contains all terms which act within one of the cubes $Q_{m}$ and $H_{S}$ contains all terms which couple at least two cubes:
  \begin{align}
    \label{eq:lattice-partition-H}
    H_{Q} &= \sum_{m \in [1:B]^{\eta}} \sum_{Z \subset Q_{m}} h_{Z},
    &
      H_{S} &= \sum_{Z \in S} h_{Z},
    &
      S = \bigcup_{i=1}^{\eta} \bigcup_{j=1}^{B-1} S_{ij}.
  \end{align}
\end{lem}

\begin{proof}
  The definition directly implies that the cubes $Q_{m}$ partition the lattice $\Lambda$ (any two cubes do not intersect and the union of all cubes equals the complete lattice).
  As the sets $Q_{m}$ are disjoint, $H_{Q}$ contains each local term from $H$ at most once.
  It remains to show that $H_{S}$ contains exactly once all local terms which are not in $H_{Q}$.
  Let $Z \subset \Lambda$ be such that $h_{Z} \ne 0$ is not in $H_{Q}$, i.e.\ there are $m, m' \in [1:B]^{\eta}$ with $m \ne m'$ such that both $Z \cap Q_{m} \ne \emptyset$ and $Z \cap Q_{m'} \ne \emptyset$.
  There is an $i \in [1:\eta]$ such that $m_{i} \ne m'_{i}$.
  Without loss of generality, assume that $m_{i} < m'_{i}$ (otherwise, exchange $m$ and $m'$).
  Let $a \in Z \cap Q_{m}$, then $a_{i} \le m_{i} \Omega$ holds.
  Let $b \in Z \cap Q_{m'}$, then $b_{i} \ge \Omega(m'_{i}-1) + 1 \ge \Omega m_{i} + 1$ holds.
  Set $j = m_{i}$, then $a \in Z \cap A_{ij}$ and $b \in Z \cap B_{ij}$ and this shows that both intersections are non-empty, i.e.\ $Z \in S_{ij} \subset S$.
  This shows that the local term $h_{Z}$, which is not in $H_{Q}$, appears in $H_{S}$ exactly once.
\end{proof}

The last \namecref{lem:lattice-partition-coupling-terms} has identified the local terms which we want to remove from $H$.
The next \namecref{lem:lattice-partition-coupling-term-location} determines the possible extent of these local terms:

\begin{lem}
  \label{lem:lattice-partition-coupling-term-location}

  Let $Z \in S_{ij}$. Then $Z \subset I_{ij} \times [1:L]^{\eta-1}$ where the interval $I_{ij} = [\Omega j - \floorofa + 1: \Omega j + \floorofa]_{i}$ is along dimension $i$ (cf.\ \cref{eq:cart-product-at-i}). 
\end{lem}

\begin{proof}
  Recall that $\diam(Z) \le a$ as $h_{Z} \ne 0$. 
  The property $Z \cap A_{ij} \ne \emptyset$
  implies
  $Z \subset B^{c}_{a}(A_{ij}) \subset C_{\floorofa}(A_{ij}) = [1:\Omega j + \floorofa]_{i} \times [1:L]^{\eta-1}$
  (\cref{lem:metric-set-diam-closed-ball,lem:metric-cube-open-ball}). 
  In the same way, $Z \cap B_{ij} \ne \emptyset$ implies
  $Z \subset B^{c}_{a}(B_{ij}) \subset C_{\floorofa}(B_{ij}) = [\Omega j - \floorofa + 1:L]_{i} \times [1:L]^{\eta-1}$. 
  Combining both provides
  $Z \subset B^{c}_{a}(A_{ij}) \cap B^{c}_{a}(B_{ij}) \subset I_{ij} \times [1:L]^{\eta-1}$ (\cref{lem:metric-cube-intersection}).
\end{proof}

The local terms $Z \in S_{ij} \subset S$, which we aim at removing, generally cover the full volume described in the last \namecref{lem:lattice-partition-coupling-term-location}; if we removed all $Z \in S_{ij}$ with a single application of \cref{lem:osborne-hamilton-perturbation}, the resulting correction $V'$ would act on a large fraction of the lattice, which we want to avoid.
In addition, a given local term $Z$ may be a member of more than one of the sets $S_{ij}$. 
We construct a partition of the set $S$ which addresses these issues:

\begin{lem}
  \label{lem:coupling-term-partition}

  Let $[1:L]^{\eta-1} = \bigdotcup_{k \in [1:B]^{\eta-1}} \tilde Q_{k}$ a partition into cubes as in \cref{eq:lattice-partition-cubes}.\footnote{%
    I.e.\  $\tilde Q_{k} = C(x_{k}, y_{k})$, $x_{k} = [\Omega (k_{i} - 1) + 1]_{i=1}^{\eta-1}$ and $y_{k} = [\Omega k_{i}]_{i=1}^{\eta-1}$.
  }
  For $i \in [1:\eta]$, $j \in [1:B-1]$ and $k \in [1:B]^{\eta-1}$, let
  \begin{align}
    \label{eq:coupling-term-partition-pre}
    S_{ijk} &= \{ Z \in S_{ij} \colon Z \cap Q_{ijk} \ne \emptyset \},
    &
      Q_{ijk} = I_{ij} \times \tilde Q_{k}
  \end{align}
  where $I_{ij}$ is from \cref{lem:lattice-partition-coupling-term-location}. 
  Then $S_{ij} = \bigcup_{k \in [1:B]^{\eta-1}} S_{ijk}$ holds and $Z \in S_{ijk}$ implies $Z \subset I_{ij} \times C_{\floorofa}(\tilde Q_{k})$.
  Subsets $S'_{ijk} \subset S_{ijk}$ which partition $S$, $S = \bigdotcup S'_{ijk}$, can be chosen in $\bigo{n^{2}}$ computational time. 
\end{lem}

\begin{proof}
  The equality $S_{ij} = \bigcup_{k \in [1:B]^{\eta-1}} S_{ijk}$ holds because the $Q_{ijk}$ partition $I_{ij} \times [1:L]^{\eta-1}$, which is a superset of all $Z \in S_{ij}$ (\cref{lem:lattice-partition-coupling-term-location});
  this equality also implies $S = \bigcup_{ijk} S_{ijk}$.

  Let $Z \in S_{ijk}$. 
  This implies $h_{Z} \ne 0$ and $\diam(Z) \le a$.
  We have $Z \subset B^{c}_{a}(Q_{ijk}) \subset C_{\floorofa}(Q_{ijk}) = C_{\floorofa}(I_{ij}) \times C_{\floorofa}(\tilde Q_{k})$ (\cref{lem:metric-set-diam-closed-ball,lem:metric-cube-open-ball}, \cref{defn:metric-cube}).
  Combining this with $Z \subset I_{ij} \times [1:L]^{\eta-1}$ provides $Z \subset [C_{\floorofa}(I_{ij}) \times C_{\floorofa}(\tilde Q_{k})] \cap [I_{ij} \times [1:L]^{\eta-1}] = I_{ij} \times C_{\floorofa}(\tilde Q_{k})$.

  In order to obtain suitable subsets $S'_{ijk} \subset S_{ijk}$, choose any fixed order for the sets $S_{ijk}$ and remove all elements from $S_{ijk}$ which are already an element of a previous $S_{ijk}$.
  This takes computational time $\bigo{n^{2}}$ where $n = L^{\eta} = \abs{\Lambda}$. 
\end{proof}

We aim at removing all interactions in a set $S'_{ijk}$ with a single application of \cref{lem:osborne-hamilton-perturbation}.
For this purpose, we define a sequence $H_{0}, \dots, H_{\Xi}$ of Hamiltonians where $H_{0} = H_{Q}$, $H_{\Xi} = H$.
Consecutive Hamiltonians in this sequence differ precisely by the local terms contained in one of the sets $S'_{ijk}$. 
In order to define this sequence of Hamiltonians, we define a specific order of the sets $S'_{ijk}$ which also proves to be advantageous below.

\begin{defn}
  \label{def:lsb-vec-function}
  For $k \in [1:B]^{\eta-1}$, let $b = \lsb(k) \in [0:1]^{\eta-1}$ be the vector whose component $b_{i}$ is the least significant bit of $k_{i}$; i.e.\ $b_{i} = 1$ ($b_{i} = 0$) if $k_{i}$ is odd (even).
\end{defn}

\begin{lem}
  \label{lem:coupling-term-order}

  Let $i \in [1:\eta]$, $j \in [1:B-1]$, $k \in [1:B]^{\eta-1}$ and $\Xi = \eta(B-1)B^{\eta-1}$.  
  Let $\omega\colon [1:\Xi] \to [1:\eta]\times[1:B-1]\times[1:B]^{\eta-1}$ be a bijective function such that its inverse $\omega^{-1}$ maps all $(i, j, k)$ with the same value of $(i, \lsb(k))$ to consecutive integers from $[1:\Xi]$.%
  \footnote{%
    For example, $\omega^{-1}(i, j, k)$ can be defined as position the of $(i, \lsb(k), j, k)$ within the lexicographically ordered sequence of all $(i, \lsb(k), j, k)$.
  }
  For $u \in [1:\Xi]$, set
  \begin{align}
    \label{eq:coupling-term-ordering-set}
    \Sigma_{u} &= \Sigma_{u-1} \dotcup S'_{\omega(u)},
    & F_{u} &= \sum_{Z \in S'_{\omega(u)}} h_{Z}
  \end{align}
  where $\Sigma_{0} = \bigdotcup_{m\in[1:B]^{\eta}} \{Z \subset Q_{m}\colon h_{Z} \ne 0\}$.
  For $u \in [0:\Xi]$ and subsets $R_{u} \subset \Lambda$, set
  \begin{align}
    \label{eq:hamiltonian-sequence}
    H_{u} &= \sum_{Z \in \Sigma_{u}} h_{Z},
    & H'_{u} &= \sum_{\substack{Z \in \Sigma_{u}\\Z\subset \bar R_{u}}} h_{Z}
  \end{align}
  Then, $H_{0} = H_{Q}$, $H_{\Xi} = H$ and $H_{u} - H_{u-1} = F_{u}$ ($u \in [1:\Xi]$).
\end{lem}

\begin{proof}
  The sets $\{Z \subset Q_{m}\colon h_{Z} \ne 0\}$ are disjoint because the sets $Q_{m}$ are disjoint (\cref{lem:lattice-partition-coupling-terms}). 
  Let $E_{H} = \{ Z \subset \Lambda \colon h_{Z} \ne 0 \}$. 
  \Cref{lem:lattice-partition-coupling-terms} implies $H_{0} = H_{Q}$ and $E_{H} = S \dotcup \Sigma_{0}$.
  $S = \bigdotcup_{u=1}^{\Xi} S'_{\omega(u)}$ is provided by in \cref{lem:coupling-term-partition} and implies $E_{H} = \Sigma_{0} \,\dotcup\, \bigdotcup_{u} S'_{\omega(u)}$, $\Sigma_{u-1}\cap S'_{\omega(u)} = \emptyset$ as well as $H_{\Xi} - H_{0} = H_{S}$, i.e.\ $H_{\Xi} = H_{Q} + H_{S} = H$. 
  $F_{u} = H_{u} - H_{u-1}$ is implied by the definitions. 
\end{proof}

The correction for removing the interactions from $S'_{\omega(u)}$ is to be supported on $\bar R_{u}$ and the choice of $R_{u} \subset \Lambda$ is still open. 
The next \namecref{lem:correction-support-overlap} defines the sets $R_{u}$ and discusses whether two given $\bar R_{u}$ overlap. 

\begin{lem}
  \label{lem:correction-support-overlap}
  
  Let $\Omega$ be an even integer and $\Omega > 4\floorofa$.
  Let $u \in [1:\Xi]$ and set
  \begin{align}
    Y_{u} &= \bigcup_{Z \in S'_{\omega(u)}} Z,
    & R_{u} &= B^{o}_{r}(Y_{u}),
    & r = \Omega/2 - 2\floorofa.
  \end{align}
  Let $i \in [1:\eta]$,\, $j, j' \in [1:B-1]$,\, $k, k' \in [1:B]^{\eta-1}$, $u = \omega^{-1}(i, j, k)$ and $u' = \omega^{-1}(i, j', k')$.
  The set $\bar R_{u}$ is at most $\bar R_{u} \subset [\Omega(j-\frac12)+1:\Omega(j+\frac12)]_{i} \times C_{\Omega/2}(\tilde Q_{k})$. 
  $\bar R_{u} \cap \bar R_{u'} = \emptyset$ holds if (i)~$j \ne j'$ or (ii)~$k \ne k'$ and $\lsb(k) = \lsb(k')$.   
\end{lem}

\begin{proof}
  \Cref{lem:coupling-term-partition} implies $Y_{u} \subset I_{ij} \times C_{\floorofa}(\tilde Q_{k})$.
  We have
  $\bar R_{u} \subset B^{c}_{a}(B^{o}_{r}(Y_{u})) \subset B^{o}_{r+a}(Y_{u}) \subset C_{r+\floorofa}(I_{ij}) \times C_{r+2\floorofa}(\tilde Q_{k})$
  and the same for $\bar R_{u'}$ and $(i, j', k')$ (\cref{lem:finite-metric-space,lem:metric-cube-open-ball,defn:metric-cube}).
  Note that $r+2\floorofa = \Omega/2$. 
  
  Assume that $j \ne j'$ holds.
  $C_{r+\floorofa}(I_{ij}) = [\Omega j - (r+2\floorofa) + 1:\Omega j + r+2\floorofa]_{i} = [\Omega(j-\frac12)+1:\Omega(j+\frac12)]_{i}$.
  This set does not intersect with the same set for $j'$ if $j \ne j'$.
  As a consequence, $\bar R_{u}$ and $\bar R_{u'}$ do not intersect (use \cref{lem:metric-cube-intersection}).

  Assume that $k \ne k'$ and $\lsb(k) = \lsb(k')$ hold.
  Let $\mu \in [1:\eta-1]$ such that $k_{\mu} \ne k'_{\mu}$.
  Without loss of generality, assume that $k_{\mu} < k'_{\mu}$ (exchange $k$ and $k'$ if necessary).
  Note that this implies $k'_{\mu} - k_{\mu} \ge 2$ because $k_{\mu}$ and $k'_{\mu}$ are both even or both odd (which follows from $\lsb(k) = \lsb(k')$).
  Note that $C_{\Omega/2}(\tilde Q_{k}) = C(x_{k} - \frac\Omega2 v, y_{k} + \frac\Omega2 v)$ where $v = (1, 1, \dots, 1)^{\eta} \in \Z^{\eta}$ and the same for $k'$. 
  We have
  \begin{align}
    [x_{k'} - \tfrac\Omega2 v]_{\mu} - [y_{k} + \tfrac\Omega2 v]_{\mu}
    &= \parena{\Omega (k'_{\mu}-1-\tfrac12)+1} - \Omega (k_{\mu}+\tfrac12) \ge 1
  \end{align}%
  where we have used $k'_{\mu} - k_{\mu} \ge 2$.
  As a consequence, $C_{\Omega/2}(\tilde Q_{k})$ does not overlap with the same set for $k'$ and this implies that $\bar R_{u}$ and $\bar R_{u'}$ do not overlap (use \cref{lem:metric-cube-intersection}).
\end{proof}

The next \namecref{lem:def-time-evo-decomp-cubic} provides the necessary definitions for applying \cref{lem:osborne-hamilton-perturbation}, taking advantage of the particular ordering function $\omega$ (\cref{lem:coupling-term-order}) and of non-overlapping sets $\bar R_{u}$ (\cref{lem:correction-support-overlap}):

\begin{lem}
  \label{lem:def-time-evo-decomp-cubic}

  Let $\Omega$ be an even integer and $\Omega > 4\floorofa$.
  For $u \in [1:\Xi]$ and $s, t \in \R$, let $V'_{us}(t)$ on $\bar R_{u}$ be the solution of $\partial_{s} V'_{us}(t) = \ii L'_{ut}(s) V'_{us}(t)$ where $L'_{ut}(s) = \tau^{H'_{u}}_{ts}(F_{u}(s))$ and $V'_{ut}(t) = \idm$. 
  Set $V'_{u} = V'_{us}(t)$ and $V' = V'_{\Xi} \dots V'_{2} V'_{1}$.
  Then, $V'$ is given by
  \begin{align}
    \label{eq:timeevo-cubic-decomp}
    V' &= \prod_{i=1}^{\eta} \prod_{l\in[0:1]^{\eta-1}} V'_{il},
    &
    V'_{il} &= \bigotimes_{j=1}^{B-1} \bigotimes_{
    \substack{
    k \in [1:B]^{\eta-1}\\\lsb(k) = l
    }
    }
    V'_{ijk}.
  \end{align}
  where $V'_{ijk} = V'_{u}$ with $u = \omega^{-1}(i, j, k)$.\footnote{%
    The order of the terms $V'_{il}$ in \eqref{eq:timeevo-cubic-decomp} is specified by the function $\omega$. 
  }
  In addition, set $V = V' U^{H_{Q}}_{ts}$.
\end{lem}

\begin{proof}
  Use \cref{lem:coupling-term-order,lem:correction-support-overlap} recalling that all $(i, j, k) = \omega(u)$ with the same value of $(i, \lsb(k))$ appear consecutively as $u$ proceeds from $1$ to $\Xi$.
\end{proof}

Finally, we have completed the preparations for applying \cref{lem:osborne-hamilton-perturbation}:

\begin{thm}
  \label{thm:time-evo-decomp-cubic}

  Let $\Lambda = [1:L]^{\eta}$, $n = \abs{\Lambda} = L^{\eta}$, let $\Omega$ be an even integer and $B = \ceil{L / \Omega}$.
  Choose $q \in (0, 1)$ and let $\Omega$ be such that $r = \Omega/2 - 2\floorofa$ satisfies $r > 0$, $\ceil{r/a} > 2\kappa + 1$ and $\ceil{r/a} \ge \frac{2\kappa}{q} \ln(\frac{\kappa}{q})$ where $\kappa = \eta - 1$.
  The distance between $V = V' U^{H_{Q}}_{ts}$ from \cref{lem:def-time-evo-decomp-cubic} and the exact time evolution $U^{H}_{ts}$ is at most 
  \begin{align}
    \label{eq:thm-time-evo-decomp-cubic-result}
    \opnorma{V' U^{H_{Q}}_{ts} - U^{H}_{ts}} \le \epsilon
  \end{align}
  if
  \begin{align}
    \label{eq:thm-time-evo-decomp-cubic-cond}
    \Omega \ge \frac{2a}{1-q} \sqba{ v \abs{t - s} + \ln\parena{\frac n\epsilon} + \ln(c_{3})}
  \end{align}
  where $c_{3} = 2 \eta a M \exp(1)/\mcz$. 

  The operator $U^{H_{Q}}_{ts}$ is the tensor product of $B^{\eta} < n$ independent time evolutions on $\Omega^{\eta}$ sites.
  The operator $V'$ consists of $\Xi = \eta (B-1)B^{\eta-1} < \eta n$ independent time evolutions on $\Omega (2\Omega)^{\eta-1}$ sites.
  All constituents of the two operators can be computed in $\bigo{n\eta \exp(\Omega^{\eta})}$ computational time.
  The operator $V$ admits a \gls{peps} representation of bond dimension $D \eqbigo{\exp(\eta 4^{\eta}\Omega^{\eta} \ln(d))}$ where $d = \max_{x\in\Lambda} d(x)$ is the maximal local dimension.
\end{thm}

\begin{proof}
  Let $u \in [1:\Xi]$ and $(i, j, k) = \omega(u)$.
  Note that $\bigdotcup_{k\in[1:B]^{\eta-1}} Q_{ijk} = I_{ij} \times [1:L]^{\eta-1}$ (cf.\ \cref{lem:coupling-term-partition}), which implies $\sum_{k} \abs{Q_{ijk}} \le 2\floorofa L^{\eta-1}$.
  The operator $F_{u}$ is the sum of a subset of all terms which intersect with $Q_{ijk}$ (\cref{lem:coupling-term-partition,eq:coupling-term-ordering-set}); i.e.\ $F_{u}$ is the sum of at most $\abs{Q_{ijk}}\mcz$ local terms.
  As a consequence, $\opnorm{F_{u}} \le (J/2) \abs{Q_{ijk}}\mcz = v \abs{Q_{ijk}}/2\ee$.
  We have $R_{u} = B^{o}_{r}(Y_{u})$ with $r = \Omega/2 - 2\floorofa$ (cf.~\cref{lem:correction-support-overlap}), therefore $d(Y_{u}, \Lambda\setminus R_{u}) / a \ge r/a = \Omega/(2a) - 2\floor a/a \ge \Omega/(2a) - 2$ (\cref{lem:finite-metric-space}).
  We have (use \cref{lem:osborne-hamilton-perturbation})
  \begin{align}
    \opnorma{
    V'_{ijk} U^{H_{u-1}}_{ts} - U^{H_{u}}_{ts}
    }
    \le
    \frac{M}{\mcz \ee} \abs{Q_{ijk}} \exp(v \abs{t-s} - (1-q)\Omega/2a + 2).
  \end{align}
  The total distance is at most the sum of such terms for all $u \in [1:\Xi]$ or all $(i, j, k)$, respectively.\footnote{%
    Completely analogous to the proof of \cref{lem:time-evo-eff-repr}.
  }
  We evaluate
  \begin{align}
    \sum_{ijk} \abs{Q_{ijk}} \le \sum_{ij} 2\floorofa L^{\eta-1} = 2\floorofa \eta (B-1) L^{\eta-1} < 2 \eta a n.
  \end{align}
  This provides
  \begin{align}
    \opnorma{V' U^{H_{Q}}_{ts} - U^{H}_{ts}} \le n c_{3} \exp(v\abs{t-s} - (1-q)\Omega/2a) \le \epsilon
  \end{align}
  where $c_{3} = 2 \eta a M \exp(1)/\mcz$.
  Note that $\abs{Q_{m}} \le \Omega^{\eta}$ and $\abs{\bar R_{u}} \le \Omega(2\Omega)^{\eta-1}$ (\cref{lem:correction-support-overlap}). 
  Using \cref{lem:peps-product-bond-dimension}, \cref{lem:peps-bond-dim-upper-bound} and \cref{eq:timeevo-cubic-decomp} shows that the bond dimension of a \gls{peps} representation of $V' U^{H_{Q}}_{ts}$ is at most $D \le \exp[(\Omega^{\eta} + \eta 2^{\eta-1} \Omega(2\Omega)^{\eta-1}) \ln(d^{2})]$. 
\end{proof}

\begin{cor}
  \label{cor:timeevo-cubic-decomp-support}
  
  Let $A$ operator which acts non-trivially on a single site $x$.
  Then $V^{\adjm} A V$ with $V = V' U^{H_{Q}}_{ts}$ acts non-trivially at most on $C_{r}(\{x\})$ and the radius $r = \lambda \Omega$ increases linearly with time and with $\ln(n/\epsilon)$ ($\lambda = \eta2^{\eta} + 1$). (Proof: Analogous to \cref{lem:op-support-growth}.)
\end{cor}

\begin{rem}
  The radius in the last \namecref{cor:timeevo-cubic-decomp-support} is proportional to $\Omega$;
  using the representation for an arbitrary lattice, this radius is proportional to $\Omega^{\eta}$ (\cref{rem:time-evo-eff-repr-thm}). 
\end{rem}

\section{Discussion%
  \label{sec:timeevo-discussion}}

In this work, we have discussed the unitary time evolution operator $U_{t}$ induced by a time-dependent finite-range Hamiltonian on an arbitrary lattice with $n$ sites.
In addition, we have discussed time-evolved states $\ket{\psi(t)} = U_{y} \ket{\psi(0)}$ where the initial state $\ket{\psi(0)}$ is a product state. 
We have shown that such a time-evolved state can be certified or verified efficiently, i.e.\ there is an efficient method to determine an upper bound $\beta$ on the infidelity of the time-evolved state $\ket{\psi(t)}$ and an arbitrary, unknown state $\rho$.
We presented a method where the measurement effort for obtaining the upper bound $\beta$ was only $\bigo{n^{3} \exp[(v\abs{t} + \ln(n/\mci))^{\eta}]}$ instead of $\bigo{\exp(n)}$.
If the time-evolved state $\ket{\psi(t)}$ and the unknown state $\rho$ are sufficiently close, the upper bound $\beta$ is guaranteed to not exceed $\mci$. 
The measurement effort is seen to increase quasi-polynomially with $n$ if the spatial dimension $\eta$ is two or larger and polynomially with $n$ in one spatial dimension.
The scaling in a single spatial dimension matches what was obtained previosly \parencite[Supplementary material]{Lanyon2016}.
The complete time evolution operator $U_{t}$ can be encoded into a time-evolved state $\ket{\psi(t)}$ if each site of the lattice is augmented by a second site of the same dimension and the initial state is one where each pair of sites is maximally entangled \parencite{Holzaepfel2014}.
A certificate for this time-evolved state then also provides a certificate for the time evolution operator $U_{t}$.
This enables assumption-free verification of the output of methods which, under the assumption that it is a finite-ranged Hamiltonian, determine the unknown Hamiltonian of a system \parencite{Silva2011,Holzaepfel2014}.

We have also shown that the time evolution operator $U_{t}$ admits an efficient \gls{pepo} representation on the same lattice as the Hamiltonian, implying that the time-evolved state $\ket{\psi(t)}$ admits an efficient \gls{peps} representation.
This holds if time $t$ is at most poly-logarithmic in the number of sites $n$.
An efficient representation on the same lattice is different from efficient \gls{pepo} representations of $U_{t}$ based on the Trotter decomposition, which use a lattice of a larger dimension than the Hamiltonian itself. 
Our result provides guidelines on the necessary resources for numerically computing the time-evolved state $\ket{\psi(t)}$ with \glspl{peps} (or a suitable subclass thereof); such methods typically attempt to represent the time-evolved state $\ket{\psi(t)}$ on the same lattice as the Hamiltonian.
We construct an efficient representation of $U_{t}$ which approximates $U_{t}$ up to an error $\epsilon$ and which is based on a unitary circuit which propagates a local observable to a region whose diameter grows only linearly with $v \abs{t} + \ln(n/\epsilon)$.
This highlights that $U_{t}$ is approximated by a \gls{pepo} with a very specific structure; a general \gls{pepo} might e.g.\ displace local observables by arbitrarily large distances.
This property can also be used for an alternative proof of efficient certification of time-evolved states $\ket{\psi(t)}$, following the original approach pursued in one spatial dimension \parencite[Supplementary material]{Lanyon2016}. 

We have shown that time-evolved states of finite-range Hamiltonians can be certified and represented efficiently.
At this point, it remains an open question whether these results can be extended to Hamiltonians with exponentially decaying couplings.


%% file: c-Appendix_TEC.tex
\section{Removing single-site terms
  \label{sec:lr-remove-single-site-terms}}

The Lieb--Robinson bounds discussed in \cref{sec:timeevo-lr} show that information propagates with a maximal velocity $v$, the Lieb--Robinson velocity, if the Hamiltonian $H(t) = \sum_{Z \subset \Lambda} h_{Z}(t)$ which governs the dynamics satisfies certain conditions.
The Lieb--Robinson velocity is given by $v = J \mcz \exp(1)$ where $J = 2\sup_{t,Z \subset \Lambda} \opnorm{h_{Z}(t)}$ is twice the maximal norm of a local term of the Hamiltonian (\cref{eq:timeevo-lr-J-a}).
Adding a term $h_{x}$ acting only on a single lattice site $x \in \Lambda$ to the Hamiltonian can increase the Lieb--Robinson velocity arbitrarily but one would not expect that it affects how information propagates in the system because it acts only on a single site.
In infinite-dimensional systems, Lieb--Robinson bounds unaffected even by unbounded single-site terms have been proven \parencite{Nachtergaele2009,Nachtergaele2010a,Nachtergaele2014}.
In the following, we provide a simple way to use \cref{thm:timeevo-quasilocality} without single-site terms influencing the Lieb--Robinson velocity.
This is achieved by switching to a suitable interaction or Dirac picture before applying the \namecref{thm:timeevo-quasilocality}.
\Cref{lem:interaction-picture} introduces the interaction picture we use and \cref{cor:interaction-picture-lr} applies it to the Lieb--Robinson bound from \cref{thm:timeevo-quasilocality}. 

\newcommand{\diracp}{\text{(D)}}
\newcommand{\dirac}{\text{D}}

The interaction or Dirac picture is introduced in most quantum mechanics textbooks and we present our version in the following \namecref{lem:interaction-picture}.
A Hamiltonian $H$ is split into two parts, $H = F + G$.
Observables $A_{\dirac}(t)$ evolve according to $F$ and states $\ket{\psi_{\dirac}(t)}$ evolve such that the correct expectation values arise.

\begin{lem}
  \label{lem:interaction-picture}

  Fix a time $r \in \R$ and let the two times $s, t \in \R$ be arbitrary. 
  Let $H(t) = F(t) + G(t)$ be a Hamiltonian and let $A(t)$ be an observable.
  Set $\ket{\psi(t)} = U^{H}_{tr}\ket{\psi(r)}$ and
  \begin{align}
    \label{eq:intpic-op-state}
      A_{\dirac}(t) &= U^{F}_{rt} \, A(t) \, U^{F}_{tr}\,,
    &
      \ket{\psi_{\dirac}(t)} &= U^{F}_{rt} \, U^{H}_{tr} \, \ket{\psi(r)}.
  \end{align}
  Expectation values are given by
  \begin{align}
    \label{eq:intpic-ept}
    \bra{\psi_{\dirac}(t)} \, A_{\dirac}(t) \, \ket{\psi_{\dirac}(t)} &= \bra{\psi(t)} \, A(t) \, \ket{\psi(t)}
  \end{align}
  Set
  \begin{align}
    \label{eq:intpic-propa-def}
    U^{\diracp}_{ts} &= U^{F}_{rt} U^{H}_{ts} U^{F}_{sr}\,.
  \end{align}
  This operator propagates the states $\ket{\psi_{\dirac}(t)}$ via
  \begin{align}
    \label{eq:intpic-propa-state}
    \ket{\psi_{\dirac}(t)} = U^{\diracp}_{ts}\ket{\psi_{\dirac}(s)}.
  \end{align}
  and it is the solution of the differential equation
  \begin{align}
    \label{eq:intpic-propa-diffeq}
    \partial_{t} U^{\diracp}_{ts} = -\ii \tilde G(t) U^{\diracp}_{ts}
  \end{align}
  where $\tilde G(t) = U^{F}_{rt} G(t) U^{F}_{tr}$ and $U^{\diracp}_{ss} = \idm$, i.e.\ $U^{\diracp}_{ts} = U^{\tilde G}_{ts}$. 
\end{lem}

\begin{proof}
  \Cref{eq:intpic-ept,eq:intpic-propa-state} follow directly from the definitions. \Cref{eq:intpic-propa-diffeq} is shown by
  \begin{align*}
    \partial_{t} U^{\diracp}_{ts}
    &=
      +\ii U^{F}_{rt} [ F(t) - H(t) ] U^{H}_{ts} U^{F}_{sr}
      =
      -\ii U^{F}_{rt} G(t) U^{F}_{tr} U^{F}_{rt} U^{H}_{ts} U^{F}_{sr}
      =
      -\ii \tilde G(t) U^{\diracp}_{ts}\,,
  \end{align*}
  which completes the proof.
\end{proof}

\begin{cor}
  \label{cor:interaction-picture-lr}
  In the setting from \cref{sec:timeevo-lr}, let $H(t) = \sum_{Z \subset \Lambda} h_{Z}(t)$, $F(t) = \sum_{x \in \Lambda} h_{\{x\}}(t)$ and $G(t) = \sum_{Z \subset \Lambda, \abs Z \ge 2} h_{Z}(t)$.
  Assume that $G(t) \ne 0$ for some $t$ and consider the parameters defined in \cref{eq:timeevo-lr-J-a,eq:timeevo-lr-Z,eq:timeevo-lr-R-n,eq:timeevo-lr-abs-R-n,eq:timeevo-lr-sites-per-term}.
  The maximal range $a$ of $H$, $G$ and $\tilde G$ is the same.
  The maximal norm $J$ satisfies $J(\tilde G) = J(G) \le J(H)$ and the same holds for the parameters $\mcz$, $M$, $\kappa$ and $\mcy$.

  Let $Y \subset R \subset \Lambda$ and let $A$ act on $Y$.
  \Cref{thm:timeevo-quasilocality} provides a bound 
  \begin{align}
    \opnorma{\tau^{H}_{ts}(A) - \tau^{H_{\bar R}}_{ts}(A)} \le \epsilon(H)
    \label{eq:interaction-picture-lr}
  \end{align}
  where $\epsilon(H)$ depends on the parameters of $H$ just mentioned, as specified \cref{thm:timeevo-quasilocality}.
  \cref{eq:interaction-picture-lr} still holds if $\epsilon(H)$ is replaced by the smaller $\epsilon(\tilde G) = \epsilon(G)$. 
\end{cor}

\begin{proof}
  $\tilde G$ and $G$ have the same value of $J$ because the operator norm is unitarily invariant.
  $\tilde G$ and $G$ have the same value of $a$, $\mcz$, $M$, $\kappa$ and $\mcy$ because the tensor product $U^{F}_{ts}$ does not change the set of sites on which a local term acts non-trivially.
  Inspection of \cref{eq:timeevo-lr-J-a,eq:timeevo-lr-Z,eq:timeevo-lr-R-n,eq:timeevo-lr-abs-R-n,eq:timeevo-lr-sites-per-term} yields claimed inequalities between parameters of $H$ and parameters of $G$.

  Applying \cref{thm:timeevo-quasilocality} to $\tilde G$ and $A'$ acting on $Y$ provides
  \begin{align}
    \label{eq:interaction-picture-lr-proof}
    \opnorma{\tau^{\tilde G}_{ts}(A') - \tau^{\tilde G_{\bar R}}_{ts}(A')}
    \le \epsilon'(\tilde G).
  \end{align}
  As $\mcz$ appears in the denominator of $\epsilon$, the claimed $\epsilon'(\tilde G) \le \epsilon(H)$ might fail to hold if $\mcz(\tilde G) < \mcz(H)$. 
  $\mcz(\tilde G) = \mcz(H)$ can be ensured by keeping arbitrarily small single-site terms in $G$ instead of removing them completely.
  If we similarly set $\kappa(G) = \kappa(H)$ and $M(G) = M(H)$, \cref{eq:timeevo-lr-abs-R-n} is satisfied for $G$. 
  Inserting $A' = \tau^{F}_{rs}(A) = \tau^{F_{Y}}_{rs}(A)$, which acts only on $Y$, into \eqref{eq:interaction-picture-lr-proof} and using the unitary invariance of the operator norm provides
  \begin{align*}
    \opnorma{\tau^{H}_{ts}(A) - \parena{\tau^{F_{\bar R}}_{tr}\tau^{\tilde G_{\bar R}}_{ts}\tau^{F_{Y}}_{rs}}(A)}
    =
    \opnorma{\tau^{F}_{tr}\parena{\tau^{\tilde G}_{ts}\parena{\tau^{F}_{rs}(A)} - \tau^{\tilde G_{\bar R}}_{ts}\parena{\tau^{F_{Y}}_{rs}(A)}}}
    \le \epsilon'(\tilde G).
  \end{align*}
  Here, we used $U^{H}_{ts} = U^{F}_{tr} U^{\tilde G}_{ts} U^{F}_{rs}$ (\cref{eq:intpic-propa-def}).
  Note that $\parenc{\tau^{F_{\bar R}}_{tr}\tau^{\tilde G_{\bar R}}_{ts}\tau^{F_{Y}}_{rs}}(A) = V_{ts} A V_{ts}^{\adjm}$ where $V_{ts} = U^{F_{\bar R}}_{tr} U^{\tilde G_{\bar R}}_{ts} U^{F_{\bar R}}_{rs}$.
  Applying \cref{lem:interaction-picture} to $H_{\bar R} = F_{\bar R} + G_{\bar R}$, where $H_{\bar R}$ was split in the same way as $H$, provides $U^{\tilde G_{\bar R}}_{ts} = U^{F_{\bar R}}_{rt} U^{H_{\bar R}}_{ts} U^{F_{\bar R}}_{sr}$, i.e.\ $U^{H_{\bar R}}_{ts} = U^{F_{\bar R}}_{tr} U^{\tilde G_{\bar R}}_{ts} U^{F_{\bar R}}_{rs} = V_{ts}$, which completes the proof. 
\end{proof}

\begin{rem}
  \label{rem:cor:interaction-picture-lr}
  Before applying \cref{cor:interaction-picture-lr}, it can be worthwhile to minimize the norm of $h_{Z}$ with $\abs Z \ge 2$ by subtracting single-site terms from it.
  These single-site terms can reduce the norm of $h_{Z}$ (i.e.\ $J$ and $v$) and they are added to the Hamiltonian as single-site terms in order to leave the total Hamiltonian unchanged. 
\end{rem}

\section{Various lemmata}

\begin{lem}
  \label{lem:unit-seq-triangle-unit-inv}
  Let $\norm\cdot$ be a unitarily invariant norm and let $U_{2}$, $V_{1}$ be unitary, $i \in \{1, 2\}$. Let $A$ be an arbitrary matrix. Then $\norm{U_{1} A U_{2} - V_{1} A V_{2}} \le \norm{(U_{1} - V_{1}) A} + \norm{A (U_{2} - V_{2})}$.
\end{lem}

\begin{proof}
  \begin{alignat*}{4}
    \norm{U_{1} A U_{2} - V_{1} A V_{2}}
    &=
      \| U_{1} A U_{2} - V_{1} A U_{2} &\;+\;& && V_{1} A U_{2} - V_{1} A V_{2} \| \\
      &\le
      \norm{U_{1} A U_{2} - V_{1} A U_{2}} &\;+\;& &\|&V_{1} A U_{2} - V_{1} A V_{2}\| \\
      &=
      \norm{(U_{1} - V_{1}) A} &\;+\;& &\| & A (U_{2} - V_{2}) \|
  \end{alignat*}
  where the triangle inequality and unitary invariance have each been used once.
\end{proof}

The following three Lemmata are used in \cref{sec:timeevo-lr}.

\begin{lem}
  \label{lem:timeevo-poly-exp-bound}
  Let $n \ge 0$, $a > 0$ and $x \ge \max\{0, \frac{2n}{a} \ln(\frac{n}{a})\}$. Then $x^{n} \exp(-ax) \le 1$. 
\end{lem}

\begin{proof}
  For $n = 0$ or $x = 0$, the Lemma holds. Let $n > 0$ and $x > 0$.
  Let $z = \frac an x$ and $c = \ln(\frac na)$.
  The inequalities $z \ge 2c$ (implied by the premise) and $\ln(z) \le \frac z2$ (see \cref{lem:timeevo-ln-x-bound}) imply $\ln(z) + c \le \frac z2 + c \le z$. 
  We have
  \begin{align}
    \ln(z) + c \le z
    \quad \Leftrightarrow \quad
    \ln(x) \le \frac{a x}{n}
    \quad \Leftrightarrow \quad
    n \ln(x) - a x \le 0
    \quad \Leftrightarrow \quad
    x^{n} \ee^{-ax} \le 1.
  \end{align}
  This completes the proof because the inequality on the very left is implied by the premise.
\end{proof}

\begin{lem}
  \label{lem:timeevo-ln-x-bound}
  $\ln(x) \le \frac x2 - (1 - \ln 2) < \frac x2$ for $x \in [0, \infty)$ with equality if and only if $x = 2$. 
\end{lem}

\begin{proof}
  Let $f(x) = \frac x2 - \ln(x) - (1 - \ln 2)$. The derivative satisfies
  \begin{align}
    f'(x)
    &=
      \frac12 - \frac 1x
      \quad\quad
      \begin{cases}
        > 0, & \text{if } x > 2, \\
        = 0, & \text{if } x = 2, \\
        < 0, & \text{if } x < 2.
      \end{cases}
  \end{align}
  In addition, $f(2) = 0$. This shows the claim.
\end{proof}

\begin{lem}
\label{lem:local-te-cert-fid-tracedist}
(i)
Let $\trnorm\cdot$ denote the trace norm, $\psi = \ketbra\psi\psi$ and $\psi' = \ket{\psi'}\bra{\psi'}$.
If $\norm{\ket{\psi} - \ket{\psi'}} \le \epsilon \le \sqrt2$, then $\trnorm{\psi - \psi'} \le 2\epsilon$.

(ii)
Let $1 - \abs{\braket{\psi}{\psi'}} = \epsilon$. Then $\min_{\alpha \in [0, 2\pi]} \norm{\ket{\psi} - \ee^{\ii \alpha}\ket{\psi'}} = \sqrt{2\epsilon}$. Let in addition $\epsilon \le 1$, then $\trnorm{\psi - \psi'} \le 2\sqrt{2\epsilon}$. 
\end{lem}

\begin{proof}
(i)
Assume that $\norm{\ket\psi - \ket{\psi'}} \le \epsilon$ holds.
This gives us
\begin{align}
  \epsilon^{2}
  &
    \ge
    \norm{\ket\psi - \ket{\psi'}}^{2}
    = 2 (1 - \Re(\braket{\psi}{\psi'})) \ge 2 (1 - \sqrt{F})
\end{align}
where $F = |\braket{\psi}{\psi'}|^{2} = F(\ket\psi, \ket{\psi'})$.
This gives $\sqrt{F} \ge 1 - \epsilon^{2}/2$ and $1 - F \le 1 - (1 - \epsilon^{2}/2)^{2} = \epsilon^{2} - \epsilon^{4}/4 \le \epsilon^{2}$. The equality $\trnorm{\psi - \psi'} = 2 \sqrt{1 - F}$ completes the proof \parencite[Eqs.~9.11, 9.60, 9.99]{Nielsen2007}. 

(ii)
Choose $\alpha \in \mathbb R$ such that, with $\ket{\psi''} = \ee^{\ii \alpha} \ket{\psi'}$, the equalities $\abs{\braket{\psi}{\psi'}} = \braket{\psi}{\psi''} = \Re(\braket{\psi}{\psi''})$ hold. In this case, we have
\begin{align}
  \min_{\alpha \in [0, 2\pi]} \norm{\ket{\psi} - \ee^{\ii \alpha}\ket{\psi'}}
  \le
  \norm{\ket{\psi} - \ket{\psi''}}^{2}
  = 2[1 - \Re(\braket{\psi}{\psi''})]
  = 2\epsilon
\end{align}
and it is clear that for all other values of $\alpha \in \mathbb R$, the value of $1 - \Re(\braket{\psi}{\psi''})$ will be larger. Part (i) proofs the remaining part of (ii). 
\end{proof}

\section{Metric spaces
  \label{sec:app-metric-spaces}}

\begin{rem}
  Given two sets $A$ and $B$, the expression $A \subset B$ is used to refer to the implication $x \in A \Rightarrow x \in B$.
\end{rem}

\begin{defn}
  \label{def:app-metric-space}
  Let $\Lambda$ be a set. A function $d \colon \Lambda \times \Lambda \to \R$ is called a metric if, for all $x, y, z \in \Lambda$, $d(x, y) \ge 0$, $d(x, y) = 0$ if and only if $x = y$, $d(x, y) = d(y, x)$ and $d(x, z) \le d(x, y) + d(y, z)$ (triangle inequality).
  The pair $(\Lambda, d)$ is called a metric space and a finite metric space is a metric space where $\Lambda$ has finitely many elements.
  Statements in this section for infinite metric spaces should be treated with caution (they are not used in the main text).

  Distances between sets are given by $d(A, B) = \inf_{a \in A, b \in B} d(a, b)$ and the infimum turns into a minimum if both sets are finite. Accordingly, we have
  \begin{subequations}
  \begin{align}
    \label{eq:app-setdist-upper-bound}
    \exists \, a_{0} \in A, b_{0} \in B \colon
    d(a_{0}, b_{0}) < r
    \quad &\Rightarrow \quad
    d(A, B) \le d(a_{0}, b_{0}) < r, \\
    \label{eq:app-setdist-lower-bound}
    \forall \, a \in A, b \in B \colon
    d(a, b) > r
    \quad &\Rightarrow \quad
    d(A, B) > r.
  \end{align}
  \end{subequations}
  Strict inequalities can be replaced by equalities in both equations.
  If the metric space is infinite, the strict inequality in the second equation turns into an inequality.

  The diameter of a subset $Y \subset \Lambda$ is given by $\diam(Y) = \sup_{x, y \in Y} d(x, y)$ and the supremum turns into a maximum for a finite set $Y$.
  Let $\mathcal M$ a set of subsets of $\Lambda$ with $a = \sup_{Z \in \mathcal M} \diam(Z) < \infty$.
  Define the extension of $R \subset \Lambda$ via $\bar R = \bigcup_{Z \in \mathcal M, Z \cap R \ne \emptyset} Z$.

  The open and closed ball around $Y \subset \Lambda$ are defined by
  \begin{subequations}
  \begin{align}
    B^{o}_{r}(Y) &= \{ x \in \Lambda \colon d(x, Y) < r \}, \\
    B^{c}_{r}(Y) &= \{ x \in \Lambda \colon d(x, Y) \le r \}.
  \end{align}
  \end{subequations}
\end{defn}

\begin{lem}
  \label{lem:finite-metric-space}
  The following hold ($Y \subset \Lambda$, $r, s \ge 0$):
  \begin{subequations}
  \begin{align}
    \label{eq:app-d-Y-o-ball-complement}
    d(Y, \Lambda \setminus B^{o}_{r}(Y)) &\ge r \\
    \label{eq:app-d-Y-c-ball-complement}
    d(Y, \Lambda \setminus B^{c}_{r}(Y)) &> r \\
    \label{eq:app-d-o-ball-complement}
    d(B^{o}_{s}(Y), \Lambda \setminus R) &> d(Y, \Lambda \setminus R) - s \\
    \label{eq:app-d-c-ball-complement}
    d(B^{c}_{s}(Y), \Lambda \setminus R) &\ge d(Y, \Lambda \setminus R) - s \\
    \label{eq:ball-subset-complement-dist}
    B^{o}_{r}(Y) &\subset R \quad\text{where}\quad r = d(Y, \Lambda \setminus R)
    \\
    \label{eq:open-ball-subset}
    \sqba{
    B^{c}_{r}(B^{o}_{s}(Y)) \cup
    B^{o}_{r}(B^{c}_{s}(Y)) \cup
    B^{o}_{r}(B^{o}_{s}(Y))} & \;\subset\;
    B^{o}_{r+s}(Y) \\
    \label{eq:app-bar-R-subset-closed-ball}
    \bar R \subset B^{c}_{a}(R)
    \\
    \label{eq:app-ball-diam}
    \diam(B^{o}_{r}(Y)) < 2r + \diam(Y).
  \end{align}
  \end{subequations}
  Strict inequalities turn into non-strict inequalities for infinite metric spaces.

  Let $x, y \in \Lambda$. Then $d(x, y) \ge r + s$ implies $B^{o}_{r}(x) \cap B^{c}_{s}(y) = \emptyset$.
\end{lem}

\begin{proof}
  For all $y \in Y$ and $z \in \Lambda \setminus B^{o}_{r}(Y)$, it is true that $z \not \in B^{o}_{r}(Y)$ and thus $d(y, z) \ge r$. \eqref{eq:app-setdist-lower-bound} thus implies \eqref{eq:app-d-Y-o-ball-complement}.

  For all $y \in Y$ and $z \in \Lambda \setminus B^{c}_{r}(Y)$, it is true that $z \not \in B^{c}_{r}(Y)$ and thus $d(y, z) > r$. \eqref{eq:app-setdist-lower-bound} thus implies \eqref{eq:app-d-Y-c-ball-complement}.
  
  Let $z \in \Lambda \setminus R$ and $x \in B^{o}_{s}(Y)$.
  There is a $y \in Y$ such that $d(x, y) < s$.
  Note that $d(y, z) \ge d(Y, \Lambda \setminus R)$ holds.
  This implies that $d(z, x) \ge d(z, y) - d(y, x) > d(Y, \Lambda \setminus R) - s$.
  \eqref{eq:app-setdist-lower-bound} implies \eqref{eq:app-d-o-ball-complement}.

  Let $z \in \Lambda \setminus R$ and $x \in B^{c}_{s}(Y)$.
  There is a $y \in Y$ such that $d(x, y) \le s$.
  Note that $d(y, z) \ge d(Y, \Lambda \setminus R)$ holds.
  This implies that $d(z, x) \ge d(z, y) - d(y, x) \ge d(Y, \Lambda \setminus R) - s$.
  \eqref{eq:app-setdist-lower-bound} implies \eqref{eq:app-d-c-ball-complement}.

  Let $x \in B^{o}_{r}(Y)$, then there is a $y \in Y$ such that $d(x, y) < r$.
  If $x \in \Lambda \setminus R$ was true, it would imply $d(Y, \Lambda \setminus R) \le d(x, y) < r$ (see~\eqref{eq:app-setdist-upper-bound}), which is a contradiction.
  Therefore, we infer $x \not\in \Lambda \setminus R$ and thus $x \in R$.
  This shows \eqref{eq:ball-subset-complement-dist}. 

  Let $x \in B^{c}_{r}(B^{o}_{s}(Y))$.
  Then there are $z \in B^{o}_{s}(Y)$ and $y \in Y$ such that
  $d(x, z) \le r$ and $d(z, y) < s$.
  This implies $d(x, y) < r + s$ and thus $x \in B^{o}_{r+s}(Y)$.
  The remaining parts of \eqref{eq:open-ball-subset} are shown in the same way.

  Let $x \in \bar R$. If $x \in R$, then $x \in B^{c}_{a}(R)$ holds.
  Let $x \in \bar R \setminus R$.
  Then there is a $Z \subset \Lambda$ such that $\diam(Z) \le a$ and $x \in Z$ and $Z \cap R \ne \emptyset$.
  Let $y \in Z \cap R$, then $d(x, y) \le \diam(Z) \le a$.
  Because $y \in R$, we can conclude $x \in B^{c}_{a}(R)$.
  This shows \eqref{eq:app-bar-R-subset-closed-ball}.

  Let $x, y \in B^{o}_{r}(Y)$.
  Then there are $x', y' \in Y$ such that $d(x, x') < r$ and $d(y, y') < r$.
  This implies $d(x, y) \le d(x, x') + d(x', y') + d(y', y) < 2r + \diam(Y)$.
  This shows \eqref{eq:app-ball-diam}. 
  
  Assume that $z \in B^{o}_{r}(x) \cap B^{c}_{s}(y)$ exists. Then $d(x, y) \le d(x, z) + d(z, y) < r + s$ contradicts the assumption.  
\end{proof}

\begin{lem}
  \label{lem:metric-cube-open-ball}

  Let $d$ be a metric with property \eqref{eq:metric-normalized}. 
  Let $x, y \in \Lambda$.
  For any $r \ge 0$, we have
  \begin{align}
    B^{c}_{r}(C(x, y)) \subset C_{s}(C(x, y)), \quad s = \floor r.
  \end{align}
\end{lem}

\begin{proof}
  Let $z \in B^{c}_{r}(C(x, y))$, then there is a $b \in C(x, y)$ such that $d(z, b) \le r$; this implies $\abs{z_{i} - b_{i}} \le r$ for all $i \in [1:\eta]$.
  In addition, $b \in C(x, y)$ implies $x_{i} \le b_{i} \le y_{i}$.
  Combining both yields $x_{i} - r \le z_{i} \le y_{i} + r$
  and this shows that $z \in C(x - su, y + su)$ where $s = \floor r$.
\end{proof}

\begin{lem}
  \label{lem:metric-set-diam-closed-ball}
  
  Let $Y, Z \subset \Lambda$ with $Y \cap Z \ne \emptyset$. 
  If $r \ge \diam(Z)$ then $Z \subset B^{c}_{r}(Y)$. 
\end{lem}

\begin{proof}
  Let $z \in Z$ and $y \in Z \cap Y$. Then $d(z, y) \le \diam(Z, Y) \le r$, i.e.\ $z \in B^{c}_{r}(Y)$.
\end{proof}

\begin{lem}
  \label{lem:metric-cube-intersection}

  Let $a, b, c, d \in \Lambda$.
  Then, $C(a, b) \cap C(c, d) = C(x, y)$ where $x_{i} = \max\{a_{i}, c_{i}\}$ and $y_{i} = \min\{b_{i}, d_{i}\}$.
\end{lem}

\begin{proof}
    $C(a, b) \cap C(c, d)
    =
    \bigtimes_{i=1}^{\eta} [a_{i}:b_{i}] \cap [c_{i}:d_{i}]
    =
    \bigtimes_{i=1}^{\eta} [x_{i}:y_{i}]
    =
    C(x, y)
    $.
\end{proof}
